\documentclass[11pt,letterpaper]{article}

\usepackage[margin=1in]{geometry}
\usepackage{microtype}

\usepackage{amsmath,amssymb,amsthm,mathtools}

\usepackage{booktabs}
\usepackage{array}
\usepackage{adjustbox}
\usepackage{enumitem}

\usepackage{algorithm}
\usepackage{algpseudocode}

\usepackage{xcolor}
\usepackage{tikz}
\usetikzlibrary{arrows.meta,positioning,calc,fit}

\usepackage{subcaption}

\usepackage[numbers,sort&compress]{natbib}
\usepackage[
    colorlinks=true,
    linkcolor=blue,
    citecolor=blue,
    urlcolor=blue
]{hyperref}
\usepackage[capitalize,nameinlink]{cleveref}

\newtheorem{theorem}{Theorem}[section]
\newtheorem{lemma}[theorem]{Lemma}

\theoremstyle{definition}

\theoremstyle{remark}

\crefname{theorem}{Theorem}{Theorems}
\Crefname{theorem}{Theorem}{Theorems}

\crefname{lemma}{Lemma}{Lemmas}
\Crefname{lemma}{Lemma}{Lemmas}

\crefname{proposition}{Proposition}{Propositions}
\Crefname{proposition}{Proposition}{Propositions}

\crefname{corollary}{Corollary}{Corollaries}
\Crefname{corollary}{Corollary}{Corollaries}

\crefname{definition}{Definition}{Definitions}
\Crefname{definition}{Definition}{Definitions}

\crefname{remark}{Remark}{Remarks}
\Crefname{remark}{Remark}{Remarks}

\crefname{algorithm}{Algorithm}{Algorithms}
\Crefname{algorithm}{Algorithm}{Algorithms}

\makeatletter
\providecommand{\theHALG@line}{}
\renewcommand{\theHALG@line}{\thealgorithm.\arabic{ALG@line}}
\makeatother

\newcommand{\Scal}{\mathcal{S}}
\newcommand{\Rcal}{\mathcal{R}}
\newcommand{\gain}{\operatorname{gain}}
\newcommand{\pos}[1]{\left[#1\right]_{+}}
\newcommand{\eps}{\varepsilon}
\newcommand{\Rfun}{\mathfrak r}
\newcommand{\Gcal}{\mathcal{G}}
\allowdisplaybreaks
\usepackage{authblk}

\title{Breaking the \(\sqrt{3}\) Barrier for Maximum Weighted 3-Set Packing}

\author[1]{%
Weitian Tong\thanks{\textbf{Corresponding author}. 
Email: \texttt{wtong.research@gmail.com}}%
}

\author[1]{%
Yao Xu
}

\affil[1]{%
\parbox{0.5\textwidth}{%
\centering
School of Computing, Georgia Southern University\\
1332 Southern Dr, Statesboro, Georgia 30458, USA
}%
}

\begin{document}
\raggedbottom
\maketitle

\begin{abstract}
We give a deterministic polynomial-time \(1.6908\)-approximation for Maximum Weighted \(3\)-Set Packing, 
breaking the \(\sqrt3\) locality-gap barrier of squared-weight local search.  
The approximation ratio for this problem progressed 
from Berman's \(2\) \citep{Berman2000} to Neuwohner's \(2-\frac{1}{63{,}700{,}992}+\eps\) \citep{Neuwohner2021}.  
Thiery and Ward then obtained \(1.786\) \citep{ThieryWard2023}, 
while Thiery subsequently improved the bound to \(1.761+\eps\) and 
finally to \(\sqrt3 \approx 1.732051 \) through a layered exchange analysis \citep{Thiery2023Thesis}.  
Thiery also proved that \(\sqrt3\) is a locality-gap lower bound for the squared-weight objective 
even with exchanges of arbitrary size.

Our algorithm performs in two phases and combines two objectives.  
Phase~I computes a bounded-exchange local optimum for the squared-weight potential 
and analyzes it through Thiery's layered framework, 
while strengthening the terminal analysis 
by preserving internal tree-edge slack for nonsingleton components 
and exploiting the incidence structure of \(3\)-sets for final singletons.
This yields an augmented structural inequality with residual positive claw gain under the original objective.
Phase~II switches to the original objective and recovers sufficient residual gain 
through an auxiliary weighted \(9\)-Set Packing instance.  
A covering argument transfers the structural bound through the high-girth lift used only in the analysis.
\end{abstract}

\clearpage

% ===============================================
\section{Introduction}
% ===============================================

For a positive integer \(k\), an instance of weighted \(k\)-Set Packing consists of a family \(\Scal\) of sets, 
each of cardinality at most \(k\), together with a weight function \(w:\Scal\to\mathbb{R}_{>0}\).  
The goal is to find a pairwise disjoint subfamily \(A\subseteq\Scal\) maximizing \(w(A)=\sum_{S\in A}w(S)\).  
For \(k=2\), the problem is polynomially equivalent to Maximum Weight Matching.  
For \(k\ge 3\), even the unweighted problem is NP-hard, 
since \(3\)-Dimensional Matching is a special case \citep{Karp1972}.  
We focus on \(k=3\), the smallest case beyond matching.

Local search is the principal algorithmic framework for weighted Set Packing and its conflict graph formulation 
\citep{ArkinHassin1998,ChandraHalldorsson2001,Berman2000,Neuwohner2021,Neuwohner2024,ThieryWard2023,Thiery2023Thesis,SingerThiery2026}.  
The \emph{conflict graph} has one vertex for each input set and an edge between two vertices precisely 
when the corresponding sets intersect.
Thus, a feasible packing is exactly an independent set in the conflict graph.  
For a current packing $A$ and an independent set $X$ to be inserted, 
the associated exchange removes the inclusive neighborhood
\(
N(X,A)=\{a\in A : a\in X \text{ or } a \text{ intersects some } x\in X\}
\)
and replaces it by $X$. Under the original objective, the gain of this exchange is $w(X)-w(N(X,A))$.

Berman's algorithm \citep{Berman2000} evaluates local exchanges 
using the \emph{squared-weight potential} \(w^2(S)=\sum_{v\in S}w_v^2\), 
rather than the original linear objective \(w(S) = \sum_{v\in S}w_v\).  
An exchange is \emph{square improving} if it strictly increases this potential, 
and a \emph{square local optimum} is a packing 
that admits no square improving exchange of the prescribed size. 
A \emph{\(t\)-claw} is an induced star with a \emph{center} and \(t\) pairwise nonadjacent leaves, 
called \emph{talons}.  
Since a set of size at most \(k\) cannot intersect \(k+1\) pairwise disjoint sets, 
the conflict graph of a \(k\)-Set Packing instance is \((k+1)\)-claw free.
By exploiting claw-shaped square improvements in \((k+1)\)-claw-free graphs, 
Berman \citep{Berman2000} obtained approximation ratio \((k+1)/2\), 
which specializes to \(2\) for weighted \(3\)-Set Packing.

Subsequent work sharpened this analysis by extracting additional structure from square local optima.  
Neuwohner obtained the first improvements beyond Berman's \((k+1)/2\) bound \citep{Neuwohner2021,Neuwohner2024}.  
Thiery and Ward showed that larger square improving exchanges yield approximation ratio at most \(1.786\) for \(k=3\) \citep{ThieryWard2023}.  
Thiery later refined the large exchange analysis to \(1.761+\eps\) for every fixed \(\eps>0\), 
and introduced a layered analysis that achieves approximation ratio \(\sqrt3 \approx 1.732051 \) for \(k=3\) \citep{Thiery2023Thesis}.  
This \(\sqrt3\) guarantee is tight for Berman's squared-weight local search 
even when arbitrarily large exchanges are allowed \citep{Thiery2023Thesis}.  
Figure~\ref{fig:sqrt3-barrier} illustrates the local configuration responsible for this barrier. 
Thus, any improvement below \(\sqrt3\) must exploit information beyond square local optimality itself.
The key point is that this obstruction is objective-specific.
A configuration may be tight for the squared-weight potential 
while still admitting positive gain under the original linear objective, 
as Figure~\ref{fig:sqrt3-barrier} illustrates.  
Our algorithm keeps the global structure certified by square local optimality 
while exposing and subsequently recovering precisely this residual linear gain.
Table~\ref{tab:comparison} summarizes the approximation results 
for weighted \(3\)-Set Packing most relevant to this development.

\begin{figure}[htbp]
    \centering
    \begin{adjustbox}{max width=.25\linewidth}
    \begin{tikzpicture}[
        current/.style={
            circle,
            draw,
            very thick,
            minimum size=12mm,
            fill=blue!8
        },
        talon/.style={
            circle,
            draw,
            thick,
            minimum size=11mm,
            fill=green!10
        }
    ]
        \node[current] (v) at (0,0) {$v$};

        \node[talon] (o1) at (-2.6,1.55) {$o_1$};
        \node[talon] (o2) at (0,1.75) {$o_2$};
        \node[talon] (o3) at (2.6,1.55) {$o_3$};

        \draw[thick] (v) -- (o1);
        \draw[thick] (v) -- (o2);
        \draw[thick] (v) -- (o3);
    \end{tikzpicture}
    \end{adjustbox}

    \caption{
    Local configuration underlying the $\sqrt{3}$ barrier for the squared-weight
    local search.  The current solution contains the center $v$ with $w(v)=1$,
    while the three pairwise nonconflicting talons satisfy
    $w(o_1)=w(o_2)=w(o_3)=1/\sqrt{3}$.
    Thus,
    $\sum_{i=1}^{3} w(o_i)^2 = w(v)^2 = 1$, so replacing $v$ by
    $\{o_1,o_2,o_3\}$ is not a strict square improvement.
    In contrast,
    $\sum_{i=1}^{3} w(o_i)=\sqrt{3}>w(v)$, so the same exchange has positive
    gain under the original objective.  Consequently, Phase~I may stop at this
    configuration, whereas Phase~II can still exploit the profitable exchange.
    }
    \label{fig:sqrt3-barrier}
\end{figure}
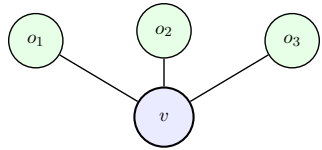

\begin{table}[htbp]
    \centering
    \caption{Approximation ratios for weighted \(3\)-Set Packing most relevant to our work.}
    \label{tab:comparison}
    \small
    \begin{adjustbox}{max width=.9\textwidth}
    \begin{tabular}{
        @{}
        >{\raggedright\arraybackslash}p{3.4cm}
        >{\raggedright\arraybackslash}p{2.2cm}
        >{\centering\arraybackslash}p{3.2cm}
        >{\raggedright\arraybackslash}p{5.8cm}
        @{}
    }
        \toprule
        \textbf{Reference}
        & \textbf{Venue}
        & \textbf{Approximation ratio}
        & \textbf{Main idea} \\
        \midrule

        Berman \citep{Berman2000}
        & SWAT 2000
        & \(2\)
        & Local search based on the squared-weight potential and claw exchanges. \\

        Neuwohner \citep{Neuwohner2021}
        & STACS 2021
        & \(2-\frac{1}{63{,}700{,}992}+\eps\)
        & Larger neighborhoods for square improving exchanges. \\

        Thiery and Ward \citep{ThieryWard2023}
        & SODA 2023
        & \(1.786\)
        & Larger square improving exchanges and a refined structural analysis. \\

        Thiery, Chapter~2 \citep{Thiery2023Thesis}
        & Ph.D. thesis, 2023
        & \(1.761+\eps\)
        & Refined analysis of large square improving exchanges. \\

        Thiery, Chapter~3 \citep{Thiery2023Thesis}
        & Ph.D. thesis, 2023
        & \(\sqrt{3}\approx 1.732051\)
        & Layered analysis of square local optima. \\

        This paper
        & Present work
        & \(\mathbf{1.6908}\)
        & Layered squared-weight local search combined with a repair phase based on gain under the original objective. \\

        \bottomrule
    \end{tabular}
    \end{adjustbox}
\end{table}

\begin{theorem}\label{thm:main}
For every weighted \(3\)-Set Packing instance with positive rational weights, 
there is a deterministic algorithm running in polynomial time that returns a packing \(A^+\) such that
\(w(O)\le 1.6908\,w(A^+)\),
where \(O\) is an optimum packing.
\end{theorem}

Our algorithm has two phases, summarized in Algorithm~\ref{alg:two phase}.
Phase~I applies Thiery's layered squared-weight local search and returns a packing \(A\).  
To describe the information retained from this phase, fix an optimum packing \(O\).  
Since single-set insertions are included among the local exchanges and all weights are positive, 
every \(o\in O\) has a nonempty inclusive neighborhood \(N(o,A)\).  
Following Berman \citep{Berman2000}, fix an arbitrary tie-breaking rule once and for all 
and assign \(o\) to the resulting maximum-weight member
\(
    \pi(o)\in\arg\max_{v\in N(o,A)} w(v),
\)
Let \(C_v=\{o\in O:\pi(o)=v\}\) for each \(v\in A\).  
Because the conflict graph is \(4\)-claw free and \(O\) is independent, \(|C_v|\le 3\).  
We define the positive gain associated with \(v\) by
\(g_v= \bigl[w(C_v)-w(N(C_v,A))\bigr]_+\),
\([x]_+:=\max\{x,0\}\),
and set
\(\Gcal(A,O)=\sum_{v\in A} g_v.\)
For the exact parameters
\(a=\frac{5433}{10000}=0.5433\) and \(c=\frac{773}{1250}=0.6184\),
our structural analysis of Phase~I proves
\[
    2w(O)\le (4-c)w(A)+a\,\Gcal(A,O).
\]
The coefficient \(4-c\) comes from the layered square-local analysis.  
The residual term \(\Gcal(A,O)\) records positive gain under the original objective 
that is not absorbed by the squared-weight analysis.

Phase~II considers every independent \(X\subseteq\Scal\setminus A\) with \(1\le |X|\le3\) 
and positive gain \(\gain_A(X):=w(X)-w(N(X,A))\).
We refer to each such \(X\) as a \emph{repair}, meaning a local replacement that could improve \(A\) under the linear objective. 
Since the ground union \(\Gamma(X):=\bigcup_{S\in X}S\) has cardinality at most \(9\), 
the candidate repairs form an auxiliary weighted \(9\)-Set Packing instance, 
where the weight of \(X\) is \(\gain_A(X)\), and two repairs are compatible if their ground unions are disjoint.
For each \(v\in A\) with \(g_v>0\), the claw \(C_v\) is such a repair, 
and the positive claws are mutually compatible 
because they consist of pairwise disjoint sets from the optimum packing \(O\).  
Thus they give a feasible auxiliary solution of total weight \(\Gcal(A,O)\).  
Applying Berman's \(5\)-approximation algorithm
\citep{Berman2000} to the auxiliary weighted \(9\)-Set Packing instance,
with scaling parameter \(\eta=1/11\), changes the factor \(5\) to \(\frac{5}{1-\eta}=\frac{11}{2}=5.5.\)
It therefore yields a compatible repair family with total declared gain at least \(\Gcal(A,O)/5.5\).
Let \(A^+\) be the packing obtained after applying these repairs and set \(\Delta_{\rm rep}:=w(A^+)-w(A)\).  
Since the actual gain is at least the total declared gain, \(\Gcal(A,O)\le 5.5\,\Delta_{\rm rep}\).  
Combining this with \(2w(O)\le(4-c)w(A)+a\Gcal(A,O)\) gives
\[
2w(O)\le(4-c)w(A)+5.5a\,\Delta_{\rm rep}\le(4-c)w(A^+),
\]
because \(5.5a=2.98815<4-c=3.3816\).  
Hence \(w(O)\le1.6908\,w(A^+)\).  
The factor inflation from \(5\) to \(11/2\) caused by scaling is therefore absorbed by the slack \(5.5a<4-c\);
in particular, improving only the auxiliary \(9\)-Set Packing routine 
would not improve the present approximation ratio.

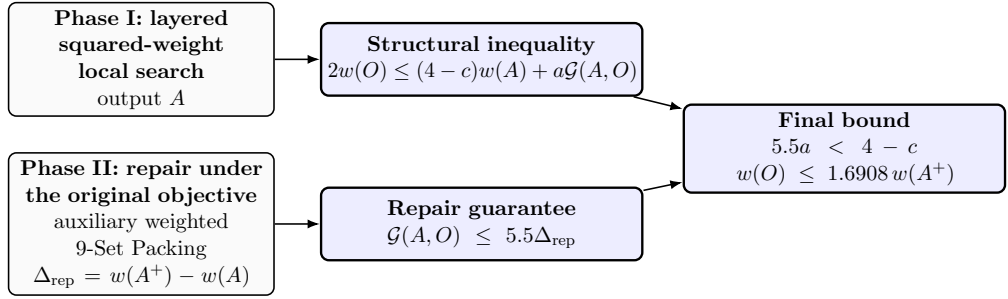
\begin{figure}[htbp]
    \centering
    \begin{adjustbox}{max width=.8\linewidth}
    \begin{tikzpicture}[
        box/.style={
            draw,
            rounded corners,
            thick,
            align=center,
            text width=4.5cm,
            minimum height=12mm,
            fill=gray!4
        },
        key/.style={
            draw,
            rounded corners,
            very thick,
            align=center,
            text width=5.5cm,
            minimum height=13mm,
            fill=blue!7
        },
        arr/.style={-{Latex[length=2.2mm]},thick}
    ]
        \node[box] (p1) at (-3.1,1.6)
        {\textbf{Phase I: layered squared-weight local search}\\
         output \(A\)};

        \node[key] (str) at (3.0,1.6)
        {\textbf{Structural inequality}\\
         \(2w(O)\le(4-c)w(A)+a\Gcal(A,O)\)};

        \node[box] (p2) at (-3.1,-1.35)
        {\textbf{Phase II: repair under the original objective}\\
         auxiliary weighted \(9\)-Set Packing\\
         \(\Delta_{\rm rep}=w(A^+)-w(A)\)};

        \node[key] (rec) at (3.0,-1.35)
        {\textbf{Repair guarantee}\\
         \(\Gcal(A,O)\le5.5\Delta_{\rm rep}\)};

        \node[key] (fin) at (9.5,0)
        {\textbf{Final bound}\\
         \(5.5a<4-c\)
         \(w(O)\le1.6908\,w(A^+)\)};

        \draw[arr] (p1) -- (str);
        \draw[arr] (p2) -- (rec);
        \draw[arr] (str) -- (fin);
        \draw[arr] (rec) -- (fin);
    \end{tikzpicture}
    \end{adjustbox}

    \caption{
        Proof architecture.  Phase~I bounds the optimum by the weight of \(A\)
        together with the residual gain \(\Gcal(A,O)\).  Phase~II recovers a
        sufficient portion of this residual gain under the original objective.
    }
    \label{fig:proof-architecture}
\end{figure}

Figure~\ref{fig:proof-architecture} summarizes the proof.  The analysis has four main ingredients.  
First, we combine the Arkin-Hassin high-girth reduction \cite{ArkinHassin1998}
with Thiery's layered tree decomposition \citep{Thiery2023Thesis}.  
We show that the residual gain \(\Gcal(A,O)\) is preserved in the required direction under the finite covering used by the reduction, 
so the structural inequality can be proved on the high-girth instance and then transferred to the original instance.  
Second, for a terminal component of Thiery's decomposition containing at least two centers, 
we retain a positive contribution from the internal tree edges 
that is discarded in the analysis for the \(\sqrt3\)-approximation algorithm. 
Third, when the terminal component consists of a single center, we use the exact multiplicities of the additional blockers.  
Each talon has size at most \(3\), and after its intersection with the assigned center is fixed, 
it can intersect at most two further members of the current packing.  
Together with square local optimality, 
this restriction reduces the final singleton case to 
a finite family of one-variable piecewise optimizations 
over the total normalized weight of the assigned talons, one for each \(m=|C_v|\in\{1,2,3\}\).
Finally, Phase~II converts the residual gain \(\Gcal(A,O)\) into an actual increase in the original objective 
by selecting a compatible family of profitable repairs.  
The underlying results of Arkin and Hassin \citep{ArkinHassin1998}, Berman \citep{Berman2000}, Thiery \citep{Thiery2023Thesis}, 
and Thiery and Ward \citep{ThieryWard2023} are used as established tools.

\paragraph{Consequences for special cases and applications.}
These include weighted \(3\)-Dimensional Matching \citep{ThieryWard2023}, 
Cycle-only Maximum Weight Kidney Exchange with cycles of length at most three \citep{BiroManloveRizzi2009}, 
Unrestricted Maximum Weight Triangle Packing \citep{HassinRubinstein2006,HassinRubinstein2006Erratum}, 
Maximum Weight \(3\)-Path Packing \citep{ZhaoXiao2025}, 
and Winner Determination with single-minded package bids of size at most three \citep{Sandholm2002}.
In each case, a feasible object is represented by the at most three vertices or items that it uses, 
with its original value as the set weight.
Thus Theorem~\ref{thm:main} yields a \(1.6908\)-approximation for each of these problems.

For weighted \(3\)-Dimensional Matching, 
Unrestricted Maximum Weight Triangle Packing, 
Bounded-size Winner Determination, 
and Bounded-size Cycle-only Maximum Weight Kidney Exchange, 
the \(1.786\)-approximation of Thiery and Ward \citep{ThieryWard2023} 
was previously inherited from weighted \(3\)-Set Packing; 
for Kidney Exchange, 
this already improved the earlier dedicated \(2+\eps\) guarantee of Bir{\'o}, Manlove, and Rizzi \citep{BiroManloveRizzi2009}.  
Thiery's dissertation subsequently improved the general weighted \(3\)-Set Packing ratio 
to \(\sqrt3\approx1.732051\) \citep{Thiery2023Thesis}. 
Maximum Weight \(3\)-Path Packing has a stronger specialized history: 
the previous ratio \(12/7\approx1.7143\) \cite{BarNoyPelegRabancaVigan2018}
was recently improved by Zhao and Xiao to \(17/10=1.7\) \citep{ZhaoXiao2025}.  
Hence our \(1.6908\) ratio improves the best previously known guarantee 
for all of these weighted special cases.

\paragraph{Organization.}
Section \ref{sec:related} reviews the hardness and approximation results most relevant to weighted Set Packing, 
with emphasis on the case \(k=3\).  
Section \ref{sec:prelim} introduces the notation and recalls the ingredients of Berman's and Thiery's analyses that are used later, 
including the slack formulation, the Arkin-Hassin high-girth reduction, and the layered exchange framework.

Section \ref{sec:algorithm} presents the two-phase algorithm and its analysis.  
Phase~I uses squared-weight local search; 
Section~\ref{sec:crossing} develops the corresponding structural analysis, including the auxiliary covering construction used to obtain tree-like exchange structures, the strengthened bounds for final components, and the resulting augmented structural inequality.
Section \ref{sec:repair} analyzes Phase~II by relating the residual positive claw gain 
to an auxiliary weighted \(9\)-Set Packing instance and bounding the gain recovered by the selected compatible repairs.  
Section \ref{sec:poly-implementation} then combines the two phase bounds to prove the \(1.6908\) approximation guarantee 
and gives the polynomial-time implementation for arbitrary positive rational weights.

Section \ref{sec:conclusion} discusses the scope of the result, its main structural bottlenecks, and directions for further work.  
Appendix~\ref{app:thiery-ideas} summarizes how local exchange interactions are organized into tree structures and subsequently refined across layers in the imported framework.
Appendix~\ref{app:numeric} records the certified numerical comparisons used in the analysis, 
and Appendix~\ref{app:finite-checks} gives the complete finite verification for the final-singleton certificate.

% ===============================================
\section{Related work}\label{sec:related}
% ===============================================

The approximation landscape for k-Set Packing differs substantially between unit and arbitrary weights.  
For unweighted instances, increasingly large structured exchanges lead to ratios close to \(k/3\), 
whereas the weighted problem has historically required nonoblivious local search and has remained closer to \(k/2\).  
We first summarize the relevant hardness results, then review the unweighted and weighted algorithmic developments, 
with particular emphasis on \(k=3\).

\paragraph{Hardness.}
Already for \(k=3\), the problem is APX-hard: Chleb\'ik and Chleb\'ikov\'a showed that Maximum \(3\)-Dimensional Matching, 
a special case of Weighted \(3\)-Set Packing, is NP-hard to approximate within \(95/94\), 
even when every element occurs in exactly two triples \citep{ChlebikChlebikova2006}.  
For general \(k\), Hazan, Safra, and Schwartz proved an \(\Omega(k/\log k)\) hardness bound unless \(\mathrm{P}=\mathrm{NP}\) \citep{HazanSafraSchwartz2006}.  
Lee, Svensson, and Thiery recently obtained hardness \(k/(12+\eps)\) for sufficiently large \(k\), 
unless \(\mathrm{NP}\subseteq\mathrm{BPP}\) \citep{LeeSvenssonThiery2025}.  
Combining their reduction with the subsequent CSP hardness of Minzer and Zheng yields 
an asymptotic \(k/8\) hardness bound under \(\mathrm{P}\ne\mathrm{NP}\) \citep{MinzerZheng2026,SingerThiery2026}.  
These large-\(k\) lower bounds remain far from the approximation ratios achievable for the fixed case \(k=3\).

\paragraph{Unweighted set packing.}
For unweighted \(k\)-Set Packing, Hurkens and Schrijver showed that 
constant-size local improvements achieve ratio \(k/2+\eps\) for every fixed \(\eps>0\) \citep{HurkensSchrijver1989}.  
Allowing logarithmic-size exchanges, Halld\'orsson obtained ratio \((k+2)/3\) in quasipolynomial time, 
and Cygan, Grandoni, and Mastrolilli improved this to \((k+1)/3+\eps\) \citep{Halldorsson1995,CyganGrandoniMastrolilli2013}.  
Polynomial-time algorithms followed: Sviridenko and Ward achieved \((k+2)/3\) 
using structured logarithmic-size improvements and color coding \citep{SviridenkoWard2013}, 
while Cygan obtained \((k+1)/3+\eps\) through bounded-pathwidth local improvements \citep{Cygan2013}.  
F\"urer and Yu obtained the same ratio with a better dependence on \(\eps\) in the running time \citep{FurerYu2014}.  
In particular, the best known polynomial-time guarantee for unweighted \(3\)-Set Packing is \(4/3+\eps\).

\paragraph{Weighted set packing.}
For weighted \(k\)-Set Packing, Arkin and Hassin proved the tight ratio \(k-1+1/t\) for direct \(t\)-local search 
using the original weights \citep{ArkinHassin1998}.  
Chandra and Halld\'orsson improved this to \(2(k+1)/3\) 
by combining greedy selection with local improvement \citep{ChandraHalldorsson2001}.  
Berman then introduced the squared-weight potential and obtained ratio \((k+1)/2\) 
through his \(d/2\)-approximation for Maximum Weight Independent Set in \(d\)-claw-free graphs \citep{Berman2000}. 
This remained the best-known approximation ratio from 2000 until 
Neuwohner gave the first improvement in 2021, obtaining \((k+1)/2-1/63{,}700{,}992+\eps\) \citep{Neuwohner2021}.
She subsequently proved that, for every fixed exponent \(\alpha\in\mathbb R\), 
local search based on \(w^\alpha\) and 
exchanges of logarithmically bounded size cannot asymptotically achieve an approximation ratio below \(k/2\); 
in the same work she obtained ratios \((k+\eps_k)/2\) with \(\eps_k\to0\) \citep{Neuwohner2024}.
Later, she combined weighted local search with unweighted subroutines, 
improving on Berman for every \(k\ge4\) and 
obtaining the asymptotic bound \(0.4986k+0.5194\) \citep{Neuwohner2023}.

For \(k=3\), Thiery and Ward obtained ratio \(1.786\) by analyzing larger square-improving exchanges \citep{ThieryWard2023}.  
Thiery's dissertation further sharpened this analysis to \(1.761+\eps\) 
and then introduced the layered exchange graph framework yielding ratio \(\sqrt3\) \citep{Thiery2023Thesis}.  
The dissertation also shows that \(\sqrt3\) is a locality gap barrier for Berman's squared-weight local search, 
even when arbitrarily large exchanges are allowed.  
% This layered analysis is the starting point of the present work; 
% our second phase exploits gain under the original objective to go below the square-potential barrier.  

For large \(k\), recent progress follows a different direction.  
Neuwohner's hybrid approach gives \(k/2.00561+O(1)\) asymptotically \citep{Neuwohner2023}, 
while Singer and Thiery recently improved this to \((\ln 4/3)k+o(k)\) 
using ordered local search together with weight bucketing \citep{SingerThiery2026}.  
These results substantially improve the asymptotic regime but do not yield a better explicit guarantee for \(k=3\).

% ===============================================
\section{Preliminaries}
\label{sec:prelim}
% ===============================================

Let $\Scal$ be a family of nonempty subsets of a universe $U$, with $|F|\le3$ for all $F\in\Scal$, and let $w:\Scal\to\mathbb{Q}_{>0}$.  
For $X\subseteq\Scal$, write $w(X)=\sum_{F\in X}w(F)$ and $w^2(X)=\sum_{F\in X}w(F)^2$.  
The conflict graph $G$ has one vertex for each set and an edge between two intersecting sets.  
We identify sets and conflict graph vertices and write $w_v=w(v)$ when a set is viewed as a graph vertex.  
For $X,Y\subseteq V(G)$, $N(X,Y)$ denotes the inclusive neighborhood of $X$ in $Y$: 
it contains $X\cap Y$ and every vertex of $Y$ adjacent to some vertex of $X$.  
We write $N(x,Y)$ for $N(\{x\},Y)$.  
Because a set of size at most three cannot meet four pairwise disjoint sets, the conflict graph is $4$-claw free.

Fix the packing \(A\) returned by Phase~I and let \(O\) be an optimum packing.  
Since Phase~I allows the insertion of a single set 
whenever this strictly increases the squared-weight potential, 
square local optimality implies that \(A\) is maximal.  
Indeed, if some \(o\notin A\) were disjoint from every set in \(A\), 
then adding \(o\) would increase \(w^2\) by \(w(o)^2>0\), contradicting square local optimality.
Hence the inclusive neighborhood $N(o,A)$ is nonempty for every $o\in O$.  
Following the fixed tie-breaking rule, 
assign each \(o\in O\) to the resulting maximum-weight member of \(N(o,A)\), 
and write \(\pi(o)\) for that center.
For $v\in A$, define the optimum claw $C_v=\{o\in O:\pi(o)=v\}$.  
The claws partition $O$ and satisfy $|C_v|\le3$.  
Put $N_v^+=\{v\}\cup\bigcup_{o\in C_v}N(o,A\setminus\{v\})$ and define the standard slack quantities
\[
\begin{aligned}
\Delta_v &=w^2(N_v^+)-w^2(C_v),\\
\Psi_v & =\sum_{o\in C_v}\Bigl((w_o-w_v)^2+w_vw(N(o,A\setminus\{v\}))-w^2(N(o,A\setminus\{v\}))\Bigr).
\end{aligned}
\]
These are the quantities used in the square-potential analyses 
of Berman \citep{Berman2000} and Thiery and Ward \citep{ThieryWard2023,Thiery2023Thesis}. 
Their roles are different. 
The quantity $\Delta_v$ is the residual square potential gap of the claw exchange centered at $v$: 
local optimality gives $\Delta_v\ge0$, and $\Delta_v=0$ means that the corresponding claw exchange is exactly tied in squared weight.  
The quantity $\Psi_v$ measures weight imbalance that is lost when the squared-weight analysis is relaxed to a linear weight estimate.  
Its term $(w_o-w_v)^2$ records a mismatch between a talon and its assigned center, 
while $w_vw(N(o,A\setminus\{v\}))-w^2(N(o,A\setminus\{v\}))$ is positive when the secondary blockers of $o$ are lighter than $v$.  
Thus $\Delta_v$ is \emph{exchange slack}, whereas $\Psi_v$ is \emph{weight separation slack}.  
For each \(v\in A\), the associated \emph{\(1\)-exchange} inserts \(C_v\) and removes \(N(C_v,A)\).  
Square local optimality implies that this exchange does not increase \(w^2\).

The inequality Eq.(\ref{eq:master}) in Lemma~\ref{lem:master} 
is the slack form of Berman's charging argument \citep{Berman2000}, 
stated explicitly as Lemma~2.3.1 in Thiery's dissertation \citep{Thiery2023Thesis}.

\begin{lemma}
\label{lem:master}
If \(A\) is square locally optimal with respect to the \(1\)-exchanges, then
\begin{equation}\label{eq:master}
2w(O)
\le
w(A)+\sum_{o\in O}w(N(o,A))
-\sum_{v\in A}
\left(
\frac{\Delta_v}{w_v}+\frac{\Psi_v}{w_v}
\right).
\end{equation}
\end{lemma}

The first two terms on the right are those of Berman's charging bound,
which yields the \((k+1)/2\)-approximation ratio.
The approximation improves when one obtains a stronger lower bound on the subtractive slack term
\(\sum_{v\in A}(\Delta_v+\Psi_v)/w_v\).  
Thiery's layered analysis strengthens this term by exploiting interactions among local exchanges; 
our analysis strengthens it further by retaining additional slack from the final components of the layered decomposition.

%----------------------------------------------------
\subsection{The Arkin-Hassin high-girth reduction}
%----------------------------------------------------

Throughout the paper, the Arkin-Hassin high-girth reduction refers to the duplication-and-crossing construction 
introduced by Arkin and Hassin \citep{ArkinHassin1998} and used by Thiery \citep[Lemma~2.5.7]{Thiery2023Thesis}.  
The Arkin-Hassin high-girth reduction is used only in the analysis.  
Its first purpose is to remove short-cycle obstructions 
so that the bounded exchange structures in Thiery's analysis can be treated as trees, 
while preserving the bounded square local optimality and the ratio \(w(O)/w(A)\) relevant to the approximation analysis.  
Its second purpose is to provide a covering projection 
under which the residual claw gain can be compared between the high-girth instance and the original instance. 
The algorithm itself never constructs the reduced graph.  
Throughout the analysis, we use the duplication-and-crossing construction 
only as a weight-preserving graph cover of \(G[A\cup O]\); 
in particular, none of the subsequent arguments requires 
the lifted graph itself to be represented as the conflict graph of a new \(3\)-set system.

The Arkin-Hassin high-girth reduction is applied to the weighted conflict graph \(G[A\cup O]\).  
In one duplication step, a disjoint copy of the current graph is created, 
with every copied vertex retaining the weight of its original vertex.  
If \((a,o)\) is an edge and \((a',o')\) is its copy, 
a crossing replaces these two edges by \((a,o')\) and \((a',o)\).  
Thus the endpoints are reconnected between the two copies without changing their weights or local degrees. 
Because \(A\) and \(O\) are packings, 
the nonisolated part of \(G[A\cup O]\) is bipartite between \(A\) and \(O\), 
and every optimum vertex has at most three neighbors in \(A\).  
The duplication-and-crossing operation preserves these local degrees.  
In particular, the lifted conflict graph retains the degree bound into the independent side \(A\)
that is used in the later \(4\)-claw-free analysis.
By applying such crossings systematically to all copies of the cycles of current minimum length, 
one eliminates those cycles without introducing shorter ones.  
Repeating this procedure raises the girth until any prescribed fixed target \(\gamma\) is reached \citep[Lemma~2.5.7]{Thiery2023Thesis}.

For every fixed \(\gamma\), the resulting finite weighted graph, denoted by \(\widetilde G\), 
has corresponding copies \(\widetilde A\) and \(\widetilde O\) of \(A\) and \(O\), has girth at least \(\gamma\), 
and preserves the bounded square local optimality required by the subsequent analysis.  
The \(4\)-claw-free property is also preserved, and
\(w(\widetilde O)/w(\widetilde A)=w(O)/w(A)\)
\citep{ArkinHassin1998}\citep[Lemma~2.5.7]{Thiery2023Thesis}.

We refer to \(\widetilde G\) as a \emph{lift} of \(G[A\cup O]\).  
A vertex \(\widetilde x\) of the lift is said to lie \emph{above} its original vertex \(x\), 
and we write \(\varpi(\widetilde x)=x\) for the natural projection back to the original graph.  
After \(r\) duplication steps, every original vertex has exactly \(M=2^r\) copies.  
Moreover, for every lifted vertex \(\widetilde x\) and every neighbor \(y\) of \(x=\varpi(\widetilde x)\), 
exactly one neighbor of \(\widetilde x\) lies above \(y\).  
Thus \(\varpi\) is locally bijective on graph neighborhoods.  
\(\widetilde G\) is therefore an \emph{\(M\)-sheet covering} of \(G[A\cup O]\), 
where the \(M\) sheets correspond to the \(M\) copies of each original vertex.  
In particular, because weights are preserved on copies, \(w(\widetilde A)=Mw(A)\) and \(w(\widetilde O)=Mw(O)\).
The duplication-and-crossing operation and the covering interpretation of the resulting lift 
are illustrated in Figure~\ref{fig:arkin-hassin-reduction}, 
provided in Appendix~\ref{app:structural-illustrations}.

The high-girth property is what permits Thiery's bounded exchange structures to be analyzed as trees.  
An exchange arc between two vertices of \(A\) is witnessed by a vertex of \(O\) 
and hence corresponds to a length-two path in \(G[A\cup O]\).  
Consequently, a sufficiently short cycle, or parallel-arc obstruction, 
in a bounded exchange structure would induce a bounded cycle in the conflict graph.  
Choosing \(\gamma\) sufficiently large rules out these obstructions, 
so the bounded connected induced exchange structures used in Thiery's analysis are trees \citep[Lemma~2.5.8]{Thiery2023Thesis}; 
see Appendix~\ref{app:thiery-trees} for the corresponding structural details.  
The covering viewpoint will additionally be used in Lemma~\ref{lem:crossing-gain} 
to transfer the required bound on the residual claw gain from the lift back to the original instance.

%----------------------------------------------------
\subsection{Thiery's exchange layers}
\label{subsec:thiery-template}
%----------------------------------------------------

We record only the notation from Thiery's layered framework \citep{Thiery2023Thesis} 
that is needed by the new proofs. 
The construction, the reason bounded exchange structures may be treated as trees, 
and the root and pendant hierarchy are explained separately in Appendix~\ref{app:thiery-ideas}.  
The definitions below are specialized to three layers and 
otherwise follow Chapter~3 of Thiery's dissertation \citep[Section~3.2]{Thiery2023Thesis}.  

Let $\eps_{-1}=0<\eps_0\le\eps_1\le\eps_2\le\eps_3<1$.  
Before applying the Arkin-Hassin high-girth reduction, 
the exchange objects are most naturally viewed as directed multigraphs $H_1,H_2,H_3,H_4$ on vertex set $A$: 
for every witness $o\in C_v$ and every secondary blocker $b\in N(o,A\setminus\{v\})$, 
we retain one directed arc from $b$ to $v$.  
Thus distinct optimum witnesses may induce parallel arcs with the same ordered endpoints.  
Such an arc belongs to $H_1$ if $w_b/w_v\in(1-\eps_1,1]$, to $H_i$ for $i\in\{2,3\}$ 
if $w_b/w_v\in(1-\eps_i,1-\eps_{i-1}]$, and to $H_4$ if $w_b/w_v\in[0,1-\eps_3]$.  
When the witness is immaterial we denote the arc simply by $(b,v)$.  
We write $H_{\le i}=H_1\cup\cdots\cup H_i$ and $H=H_{\le4}$, 
and we ignore orientations when discussing connected components.  
The Arkin-Hassin high-girth reduction
eliminates the bounded-length cycles and parallel edge obstructions needed for the subsequent induced tree decomposition.

Fix an integer $\ell\ge2$ and put $I_0=A$.  
We recursively partition the vertices as follows \citep[Definitions~3.2.3--3.2.4]{Thiery2023Thesis}.  
At layer $1$, $D_1$ consists of the vertices in connected components of $H_1$ of size at least $\ell$, 
and $I_1=A\setminus D_1$.  
For $i\in\{2,3\}$, consider the induced graph $H_{\le i}[I_{i-1}]$.  
The set $D_i$ consists of the vertices lying in its connected components of size at least $\ell$.  
Among its components of size at most $\ell-1$, the set $P_i$ contains those vertices 
whose component is contained in a larger connected component of the full graph $H_{\le i}$; 
the remaining vertices form $I_i$.  
Thus $I_{i-1}=D_i\mathbin{\dot\cup}P_i\mathbin{\dot\cup}I_i$,
where \(\mathbin{\dot\cup}\) denotes disjoint union;
that is, \(D_i\), \(P_i\), and \(I_i\) are pairwise disjoint.
Define
\[
S_1=D_1,\qquad S_2=D_2\mathbin{\dot\cup}P_2,\qquad S_3=D_3\mathbin{\dot\cup}P_3,\qquad S_4=I_3.
\]
Hence $A=S_1\mathbin{\dot\cup}S_2\mathbin{\dot\cup}S_3\mathbin{\dot\cup}S_4$.  
These are the layer classes that enter the charging argument below.
Figure~\ref{fig:thiery-exchange-layers} in Appendix~\ref{app:structural-illustrations}
illustrates this recursive partition on a fixed exchange graph.

After the Arkin-Hassin high-girth reduction, 
Corollary~3.2.8 of Thiery decomposes $A$ into vertex disjoint induced trees.  
Each tree has a root part contained in one of $D_1,D_2,D_3,I_3$ 
and possibly pendant parts contained in later $P_i$ classes.  
In particular, the final root pieces are exactly the connected components $T\subseteq I_3$, 
and each has size at most $\ell-1$ \citep[Corollary~3.2.8 and Definition~3.2.9]{Thiery2023Thesis}. 
We use this decomposition as an imported structural result.
For auditability, Appendix~\ref{app:thiery-ideas} records the precise
specializations and hypotheses from Thiery's framework that are invoked
in the proof below, together with the corresponding source lemmas.
We do not reconstruct the general tree-construction procedure.

For any tree component $T$ in this decomposition, define the \emph{layer credit}
\[
Q(T)=\sum_{v\in T}\sum_{o\in C_v}\sum_{i=1}^{3}\eps_{i-1}\,w\bigl(N(o,S_{i+1})\bigr).
\]
The role of $Q(T)$ is bookkeeping: in Thiery's proof of Theorem~3.7.1, 
these credits cancel the neighborhood term in the master inequality Eq.(\ref{eq:master})
when one sums over all decomposition trees.

For \(k=3\), write \(\rho_3(\eps)=3\eps-\eps/(1-\eps)\).  
Thiery associates a slack coefficient with each root or pendant part of the tree decomposition.  
For a root tree \(T\subseteq D_{j+1}\) of size \(t\), with \(j\ge1\), 
the coefficient is \(\chi_{j-1}^{(t)}=\rho_3(\eps_{j-1})+\frac{t-1}{t}(1-\eps_{j-1})^5\); 
for a first-layer root tree \(T\subseteq D_1\), it is \(\chi_{-1}^{(t)}=\frac{t-1}{t}(1-\eps_0)^2\).  
For a pendant tree \(T\subseteq P_{j+1}\), 
the corresponding coefficient is 
\(\nu_{j-1}=\rho_3(\eps_{j-1})+(1-\eps_{j-1})^5[1-(1-\eps_{j-1})^2\eps_{j-1}]\) 
\citep[Lemmas~3.5.1, 3.5.2, and~3.6.1]{Thiery2023Thesis}.  
These coefficients measure the effective slack contributed, per unit weight, 
by the corresponding root or pendant part after the layer-dependent terms are aggregated in the final analysis. 
Consequently, the minimum applicable coefficient controls the structural improvement 
in Thiery's global bound \citep[Theorem~3.7.1]{Thiery2023Thesis}.  
We use these imported inequalities with our parameter choices 
and strengthen only the treatment of the final components.  
The additional final-root estimate needed for this strengthening is derived in Section~\ref{sec:crossing}, 
after the high-girth lift and the component structure used in its proof have been introduced.

The next lemma isolates the final-layer slack estimate used in our analysis.  
The statement itself is not given explicitly in Thiery's dissertation \citep{Thiery2023Thesis}; 
it follows from the tree-slack and layered-slack inequalities proved there.
For \(R\subseteq A\), define \(C_R=\bigcup_{v\in R}C_v\) and \(N_R^+=\bigcup_{v\in R}N_v^+\). 

\begin{lemma}\label{lem:final-root-slack}
Let \(T\subseteq I_3=S_4\) be a final root component in Thiery's decomposition after the Arkin-Hassin high-girth reduction.  
Suppose that every connected subtree \(R\) of \(H[T]\) is square locally optimal, that is, \(w^2(N_R^+)\ge w^2(C_R)\).  
Then
\[
\sum_{v\in T}
\left(
\frac{\Delta_v}{w_v}
+
\frac{\Psi_v}{w_v}
\right)
\ge
Q(T)
+
\bigl(\rho_3(\eps_2)-3\eps_2\bigr)w(T)
+
(1-\eps_2)\sum_{(u,v)\in E(H[T])}w_u .
\]
\end{lemma}

\begin{proof}
Since \(T\subseteq I_3\) is a final root component, it is a tree of the decomposition with no pendant part.  
By Thiery's Corollary~3.2.8, the induced exchange graph \(H[T]\) is a tree.

For a set \(E'\) of directed exchange arcs, let
\(
\Gamma(E')=\sum_{(u,v)\in E'}\frac{w_u^2}{w_v}.
\)
Thiery's tree-slack inequality \citep[Lemma~2.4.5]{Thiery2023Thesis}, 
applied through the specialization in \citep[Lemma~3.4.2]{Thiery2023Thesis}, gives
\(
\sum_{v\in T}\frac{\Delta_v}{w_v}
\ge
\Gamma(E(H[T])).
\)
The hypothesis that every connected subtree of \(H[T]\) is square locally optimal is 
exactly the local-optimality condition used in the proof of this tree-slack bound.

We next specialize Thiery's layered slack calculation to the present final component.  
Apply the proof of \citep[Lemma~3.4.6]{Thiery2023Thesis} with \(k=3\), \(j=L=3\), and \(Y=\varnothing\).  
Under this specialization, the boundary arc families appearing in that proof are empty.  
The coefficient of \(w(T)\) becomes \(\rho_3(\eps_2)-3\eps_2\),
the internal-edge contribution has coefficient \(1-\eps_2\), 
and the remaining layer-neighborhood term is
\(
\sum_{v\in T}\sum_{o\in C_v}
\sum_{i=1}^{3}\eps_{i-1}w\bigl(N(o,S_{i+1})\bigr).
\)
Combining the exterior-slack estimate of \citep[Lemma~3.4.3]{Thiery2023Thesis} 
with the interior-slack estimate of \citep[Lemma~3.4.5]{Thiery2023Thesis} therefore gives
\[
\begin{aligned}
\Gamma(E(H[T]))
+\sum_{v\in T}\frac{\Psi_v}{w_v}
\ge\;&
\bigl(\rho_3(\eps_2)-3\eps_2\bigr)w(T) 
+(1-\eps_2)\sum_{(u,v)\in E(H[T])}w_u \\
&+\sum_{v\in T}\sum_{o\in C_v}
\sum_{i=1}^{3}\eps_{i-1}w\bigl(N(o,S_{i+1})\bigr).
\end{aligned}
\]
By the definition of the layer credit in Section~\ref{subsec:thiery-template}, 
the final triple sum is exactly \(Q(T)\).
Combining this inequality with
\(\sum_{v\in T}\Delta_v/w_v\ge\Gamma(E(H[T]))\) proves the claim.
\end{proof}

The final term in Lemma~\ref{lem:final-root-slack} is the part of Thiery's calculation that we retain more carefully.  
In the proof of Thiery's Theorem~3.7.1, the final-root contribution is ultimately bounded by its worst case at \(|T|=1\); 
in that case \(E(H[T])=\varnothing\), so the internal-edge term disappears.  
If \(|T|\ge2\), however, \(H[T]\) is a tree with at least one edge.  
Our analysis retains this positive contribution and combines it with the weight relations 
specific to the final layer to obtain a stronger bound for every nonsingleton final component.  
The bounds for all nonfinal root and pendant components are taken unchanged from Thiery's analysis.

% ===============================================
\section{Our algorithm}
\label{sec:algorithm}
% ===============================================

The conceptual two-phase procedure is given in Algorithm~\ref{alg:two phase}.
Phase~I computes a packing satisfying the square-local-optimality condition required by Thiery's layered analysis.  
Phase~II then switches from the squared-weight potential to the original linear objective 
and searches for a compatible collection of small positive-gain repairs.  
The optimum solution \(O\), the claw partition, the exchange layers, 
and the Arkin-Hassin high-girth construction are used only in the analysis and are never constructed by the algorithm.

For the structural analysis, we take \(L=3\) and \(\ell=7\), and set
\(a=5433/10000=0.5433\), \(\delta=1-a=0.4567\), and
\(c=773/1250=0.6184\).  The layer thresholds introduced above are specialized by
\(1-\eps_0=a^{1/4}\), \(1-\eps_1=\sqrt a\), \(1-\eps_2=a\), and
\(1-\eps_3=a^2\); in particular, \(\eps_2=\delta\).  These identities satisfy Thiery's recurrence
\(1-\eps_i=(1-\eps_{i-1})^2\).
The numerical values of \(a\) and \(c\) were selected so that the weakest imported root and pendant coefficients 
and the new final-singleton coefficient all remain strictly above \(c\).  
The subsequent proof uses only the certified inequalities in
Table~\ref{tab:coefficients}, Lemma~\ref{lem:singleton-certificate} and
Appendix~\ref{app:numeric}, Appendix~\ref{app:finite-checks}; 
the numerical search used to locate these parameters is not part of the correctness argument.

\paragraph{Phase I: square local search.}
Following Thiery, for an integer \(s\ge1\) 
we call an independent family
\(X\subseteq\Scal\setminus A\) with \(|X|\le 3s\) an \(s\)-exchange for the current packing \(A\).  
Applying the exchange replaces \(N(X,A)\) by \(X\), 
producing \((A\setminus N(X,A))\cup X\).  
The exchange is square improving when \(w^2(X)>w^2(N(X,A))\).

Thiery's Theorem~3.7.1 of \citep{Thiery2023Thesis} 
requires local optimality with respect to \(s\)-exchanges for
\[
s\ge (k^2\ell)^L\sum_{r=0}^{2\ell-2}(k\ell)^{2r}.
\]
Specializing to \(k=3\), \(L=3\), and \(\ell=7\), we therefore set
\(
s_0=(9\ell)^L\sum_{r=0}^{2\ell-2}(3\ell)^{2r}
   =63^3\sum_{r=0}^{12}21^{2r}.
\)
Phase~I initializes \(A\) greedily in nonincreasing order of weight 
and then repeatedly applies square-improving \(s_0\)-exchanges until none remains. 
The resulting packing \(A\) therefore satisfies exactly the local-optimality hypothesis 
used in Thiery's layered analysis \citep[Algorithm~1 and Theorem~3.7.1]{Thiery2023Thesis}.  
Although this conceptual local search terminates 
because the square potential increases strictly at every step, 
arbitrary rational weights do not by themselves give a polynomial bound on the number of iterations.  
The polynomial implementation is supplied later 
by the scaling and partial-enumeration wrapper 
in Section~\ref{sec:poly-implementation} and Algorithm~\ref{alg:poly-wrapper}.

\paragraph{Phase II: compatible repairs under the linear objective.}
For every independent family \(X\subseteq\Scal\setminus A\) with
\(1\le |X|\le3\), recall that
\(\gain_A(X)=w(X)-w(N(X,A))\).
We consider only candidates with \(\gain_A(X)>0\).  Their ground union
\(\Gamma(X)=\bigcup_{F\in X}F\) has cardinality at most \(9\).

If two candidates have disjoint ground unions, 
then all sets added by the two repairs are mutually disjoint, 
so the repairs can be applied simultaneously.  
More generally, for any collection \(\mathcal X\) with pairwise disjoint ground unions, 
the simultaneous repair is feasible and its actual gain is at least \(\sum_{X\in\mathcal X}\gain_A(X)\), 
since a set of \(A\) removed by several repairs is charged only once.  
Thus, choosing a compatible family of repairs reduces to a weighted \(9\)-Set Packing instance: 
the auxiliary set associated with \(X\) is \(\Gamma(X)\), with weight \(\gain_A(X)\).  
If several candidates have the same ground union, we retain one of maximum gain, 
since they have the same compatibility relation with every other repair.

The conflict graph of a \(9\)-Set Packing instance is \(10\)-claw free.  
Hence Berman's \(d/2\)-approximation for Maximum Weight Independent Set in \(d\)-claw-free graphs 
gives factor \(5\) for a squared-weight local optimum of the auxiliary instance \citep{Berman2000}.  
To obtain a polynomial number of local-search iterations on rational auxiliary weights, 
we use the standard scaling argument \citep[Section~2.5.3]{Thiery2023Thesis} 
with accuracy parameter \(\eta=1/11\).  
An exact-locality factor \(\alpha\) becomes \(\alpha/(1-\eta)\) under this scaling; 
therefore the auxiliary routine has approximation ratio
\(5/(1-1/11)=11/2=5.5\).

\begin{algorithm}[t]
\caption{Conceptual two-phase core for weighted \(3\)-Set Packing}
\label{alg:two phase}
\begin{algorithmic}
\Statex \textbf{Input:} a family \(\Scal\) of sets of cardinality at most \(3\), 
    with positive rational weights \(w\)
\Statex \textbf{Output:} a packing \(A^+\)

\State Initialize \(A\) by greedily adding sets in nonincreasing order of weight.
\While{there exists a square-improving \(s_0\)-exchange \(X\)}
    \State \(A\gets (A\setminus N(X,A))\cup X\).
\EndWhile

\State \(\Rcal\gets\varnothing\).
\For{each independent \(X\subseteq\Scal\setminus A\) with \(1\le |X|\le3\)}
    \If{\(\gain_A(X)>0\)}
        \State add the auxiliary set \(\Gamma(X)\), with weight \(\gain_A(X)\), to \(\Rcal\);
        \Statex \hspace{\algorithmicindent} if \(\Gamma(X)\) is already represented, retain only a maximum-gain representative.
    \EndIf
\EndFor

\State Apply the scaled Berman routine to \(\Rcal\) with \(\eta=1/11\), 
    and let \(\mathcal X\) be the repair candidates represented 
    by the selected pairwise disjoint auxiliary sets.

\State \(A^+\gets
\bigl(A\setminus\bigcup_{X\in\mathcal X}N(X,A)\bigr)
\cup\bigcup_{X\in\mathcal X}X\).
\State \Return \(A^+\).
\end{algorithmic}
\end{algorithm}

The remainder of the proof analyzes the two phases of the algorithm 
and then establishes the polynomial-time implementation.  
Section~\ref{sec:crossing} analyzes Phase~I.  
The Arkin-Hassin high-girth construction is used only as an analysis device: 
the layered component bounds are established on the lifted instance, 
the residual positive claw gain is controlled under the covering projection, 
and the resulting structural inequality is transferred back to the original instance.  
Section~\ref{sec:repair} analyzes Phase~II: 
the positive-gain optimum claws provide a feasible benchmark 
for the auxiliary weighted \(9\)-Set Packing instance, 
and the selected compatible repairs recover 
a constant fraction of this benchmark under the original weight objective.  
Finally, Section~\ref{sec:poly-implementation} combines 
the Phase~I structural bound with the Phase~II repair guarantee to obtain the \(1.6908\)-approximation 
and gives the scaling and partial-enumeration argument needed 
for a deterministic polynomial-time implementation with arbitrary positive rational weights.

% ===============================================
\subsection{Analysis of Phase I}
\label{sec:crossing}
% ===============================================

Recall that, for each \(v\in A\),
\(g_v=\bigl[w(C_v)-w(N(C_v,A))\bigr]_+\),
\(\Gcal(A,O)=\sum_{v\in A}g_v.\)
The quantity \(\Gcal(A,O)\) is not used by the algorithm.  
It measures the positive linear gain of the optimum claws that remains after the square-potential analysis 
and will serve as the interface with Phase~II.

The layered analysis is carried out after applying the Arkin-Hassin high-girth reduction described above.  
The standard reduction preserves the relevant bounded square local optimality and the ratio \(w(O)/w(A)\) \citep{ArkinHassin1998} 
and is used in this form by Thiery \citep[Lemma~2.5.7]{Thiery2023Thesis}.  
What is additionally needed here is that the residual gain \(\Gcal(A,O)\) does not increase, 
after normalization by the number of sheets, when we pass to the high-girth instance.

There is one harmless special case.  If \(x\in A\cap O\), then \(x\) is isolated in \(G[A\cup O]\), 
because both \(A\) and \(O\) are packings.  Hence \(N(x,A)=\{x\}\), \(C_x=\{x\}\), and \(g_x=0\).  
Such vertices are duplicated by the high-girth construction but are incident to no crossing edge.  
All nontrivial crossings therefore occur on the bipartite conflict graph between \(A\setminus O\) and \(O\setminus A\).

The next lemma extracts the covering property of the duplication-and-crossing construction that is needed later.

\begin{lemma}\label{lem:crossing-gain}
Let \((\widetilde G,\widetilde A,\widetilde O)\)
be the weighted \(M\)-sheet covering obtained from \(G[A\cup O]\) produced 
by the Arkin-Hassin high-girth reduction, and let
\(\varpi:\widetilde A\cup\widetilde O\to A\cup O\)
be the covering projection.
There is a maximum-weight center assignment
\(\widetilde\pi:\widetilde O\to\widetilde A\) such that the following hold for every
\(\widetilde v\in\widetilde A\), where \(v=\varpi(\widetilde v)\).

\begin{enumerate}[leftmargin=*,label=(\alph*)]
\item
If
\(
\widetilde C_{\widetilde v}
=
\{\widetilde o\in\widetilde O:
  \widetilde\pi(\widetilde o)=\widetilde v\},
\)
then the restriction of \(\varpi\) to
\(\widetilde C_{\widetilde v}\) is a weight-preserving bijection onto \(C_v\).

\item
The lifted claw neighborhood projects exactly onto the original claw neighborhood:
\(
\varpi\!\left(
N(\widetilde C_{\widetilde v},\widetilde A)
\right)
=
N(C_v,A).
\)
Consequently,
\(w(N(\widetilde C_{\widetilde v},\widetilde A))
 \ge w(N(C_v,A))\).
Thus, defining
\(
\widetilde g_{\widetilde v}
=
\bigl[
w(\widetilde C_{\widetilde v})
-
w(N(\widetilde C_{\widetilde v},\widetilde A))
\bigr]_+,
\)
we have
\(\widetilde g_{\widetilde v}\le g_v\).

\item
For every
\(\widetilde o\in\widetilde O\setminus\widetilde A\),
the projection \(\varpi\) maps the graph neighborhood of
\(\widetilde o\) in \(\widetilde A\) bijectively onto the graph neighborhood of
\(o=\varpi(\widetilde o)\) in \(A\).
In particular, \(\widetilde o\) has at most three neighbors in
\(\widetilde A\); after its assigned center is fixed, it has at most two secondary blockers.
\end{enumerate}

Therefore, for every
\(\widetilde J\subseteq\widetilde A\),
\[
\sum_{\widetilde v\in\widetilde J}
\widetilde g_{\widetilde v}
\le
M\,\Gcal(A,O).
\]
\end{lemma}

\begin{proof}
We first record the covering property of the Arkin-Hassin construction.  
In one duplication step, two identical copies of the current graph are created.  
For a selected edge \((a,o)\) and its copy \((a',o')\), 
the construction deletes these two edges and inserts the crossed pair \((a,o')\) and \((a',o)\) \citep{ArkinHassin1998,Thiery2023Thesis}.  
Thus each copy of \(a\) remains incident to exactly one edge lying above \((a,o)\), and the same is true for each copy of \(o\).  
All other lifted edges are unchanged.

It follows that the graph after one duplication-and-crossing step is a two-sheet covering of the graph before that step: 
for every lifted vertex and every neighbor of its projection, there is exactly one neighboring lift of that neighbor.  
Repeating the construction preserves this property by composition of covering maps.  
After \(r\) duplication steps, the resulting graph is therefore an \(M=2^r\)-sheet covering of the original graph.

We now define the lifted assignment.  Let
\(\widetilde o\in\widetilde O\setminus\widetilde A\) lie above \(o\in O\setminus A\), 
and let \(v=\pi(o)\).  Since \(v\in N(o,A)\), the covering property gives a unique neighbor of \(\widetilde o\) lying above \(v\); 
define this vertex to be
\(\widetilde\pi(\widetilde o)\).
All copies preserve their original weights, and \(v\) is a maximum-weight member of \(N(o,A)\).  
Hence \(\widetilde\pi(\widetilde o)\) is a maximum-weight member of \(N(\widetilde o,\widetilde A)\).
If \(o\in A\cap O\), then every copy \(\widetilde o\) belongs to both
\(\widetilde A\) and \(\widetilde O\); in this case set
\(\widetilde\pi(\widetilde o)=\widetilde o\).

Fix \(\widetilde v\in\widetilde A\) above \(v\in A\).
For each \(o\in C_v\setminus A\), local bijectivity at \(\widetilde v\) gives exactly one copy \(\widetilde o\) of \(o\)
adjacent to \(\widetilde v\), and by construction this is precisely the copy assigned to \(\widetilde v\).
If \(v\in A\cap O\), its unique corresponding copy in the claw is \(\widetilde v\) itself.
Conversely, any vertex assigned to \(\widetilde v\) projects to an \(o\in O\) satisfying \(\pi(o)=v\).
Thus \(\varpi\) restricts to a bijection \(\widetilde C_{\widetilde v}\to C_v\).
Since vertex weights are preserved by the construction, \(w(\widetilde C_{\widetilde v})=w(C_v)\).
This proves~(a).

For~(b), consider first a noncommon \(\widetilde o\in\widetilde C_{\widetilde v}\), 
and write \(o=\varpi(\widetilde o)\).
Local bijectivity gives
\(\varpi\!\left(N(\widetilde o,\widetilde A)\right)=N(o,A)\),
because in this case the inclusive neighborhoods coincide with the ordinary graph neighborhoods.
If \(o\in A\cap O\), then \(o\) is isolated in \(G[A\cup O]\), 
so \(N(o,A)=\{o\}\); similarly \(N(\widetilde o,\widetilde A)=\{\widetilde o\}\).
The same equality therefore holds in the common-vertex case.

Using part~(a) and the definition of inclusive neighborhood,
\[
\begin{aligned}
\varpi\!\left(
N(\widetilde C_{\widetilde v},\widetilde A)
\right)
&=
\varpi\!\left(
\bigcup_{\widetilde o\in\widetilde C_{\widetilde v}}
N(\widetilde o,\widetilde A)
\right) 
=
\bigcup_{o\in C_v}N(o,A)
=
N(C_v,A).
\end{aligned}
\]
The projection on the entire claw neighborhood need not be injective: 
distinct lifted talons may have different neighbors that project to the same original blocker.  
Nevertheless, every vertex of \(N(C_v,A)\) has at least one weight-preserving preimage in
\(N(\widetilde C_{\widetilde v},\widetilde A)\).
Since all weights are positive,
\(
w(N(\widetilde C_{\widetilde v},\widetilde A))
\ge
w(N(C_v,A)).
\)
Combining this inequality with \(w(\widetilde C_{\widetilde v})=w(C_v)\) 
and taking positive parts yields \(\widetilde g_{\widetilde v}\le g_v\).

Part~(c) follows directly from local bijectivity of the covering.
For \(o\in O\setminus A\), the sets in \(A\) are pairwise disjoint and \(o\) has cardinality at most three.  
Hence \(o\) can intersect at most three members of \(A\).  
The covering preserves this degree exactly, 
so every \(\widetilde o\in\widetilde O\setminus\widetilde A\) has at most three neighbors in \(\widetilde A\).
One of them is its assigned center, leaving at most two secondary blockers.

Finally, every \(v\in A\) has exactly \(M\) copies in \(\widetilde A\), 
and part~(b) gives \(\widetilde g_{\widetilde v}\le g_{\varpi(\widetilde v)}\) for every such copy.  
Therefore
\[
\begin{aligned}
\sum_{\widetilde v\in\widetilde J}
\widetilde g_{\widetilde v}
&\le
\sum_{\widetilde v\in\widetilde A}
g_{\varpi(\widetilde v)}
=
M\sum_{v\in A}g_v 
=
M\,\Gcal(A,O),
\end{aligned}
\]
which proves the final claim.
\end{proof}

We carry out the component analysis on the high-girth lift and suppress tildes, 
using the same notation for the corresponding lifted objects.  
By the choice of target girth in Thiery's decomposition, 
the bounded connected induced exchange subgraphs used below are trees 
\citep[Lemma~2.5.8 and Corollary~3.2.8]{Thiery2023Thesis}. 
We return to the original instance when the component bounds are combined in Section~\ref{sec:structural}.

\subsubsection{Nonfinal and nonsingleton components}

Set \(\kappa_i=1-\eps_i\) for \(i=0,1,2,3\).  
Under the parameters fixed in Section~\ref{sec:algorithm},
\(\kappa_0=a^{1/4}\), \(\kappa_1=a^{1/2}\), \(\kappa_2=a\), and \(\kappa_3=a^2\).
For convenience, define \(\Rfun(t)=4-3t-t^{-1}\).
Since \(\rho_3(\eps)=3\eps-\eps/(1-\eps)\), we have \(\Rfun(t)=\rho_3(1-t)\).

Thiery's estimates for ordinary root trees, first-layer root trees, 
and pendant trees are given in Lemmas~3.5.1, 3.5.2, and~3.6.1, respectively in \citep{Thiery2023Thesis}.  
Specializing those formulas to \(L=3\), \(\ell=7\), 
and the thresholds above gives the first five rows of Table~\ref{tab:coefficients}.  
The last two rows are the strengthened bounds for nonsingleton final components established below in Lemma~\ref{lem:final-nonsingleton}.

\begin{table}[htbp]
\centering
\caption{Structural coefficients for \(a=0.5433\) and \(\ell=7\).
The first five rows specialize Thiery's root and pendant bounds; 
the last two rows are the strengthened final-component bounds established below.  
Every coefficient exceeds \(c=0.6184\).}
\label{tab:coefficients}
\small
\begin{adjustbox}{max width=.6\textwidth}
\begin{tabular}{@{}lll@{}}
\toprule
component type
&
coefficient
&
certified lower bound
\\
\midrule
first-layer root
&
\(\displaystyle
\chi_{-1}^{(7)}
=\frac67\kappa_0^2
\)
&
\(>0.631790458675\)
\\
ordinary root, level \(0\)
&
\(\displaystyle
\chi_0^{(7)}
=\Rfun(\kappa_0)+\frac67\kappa_0^5
\)
&
\(>0.659423133441\)
\\
ordinary root, level \(1\)
&
\(\displaystyle
\chi_1^{(7)}
=\Rfun(\kappa_1)+\frac67\kappa_1^5
\)
&
\(>0.618533561989\)
\\
pendant, level \(0\)
&
\(\displaystyle
\nu_0
=\Rfun(\kappa_0)
+\kappa_0^5
 \bigl[1-\kappa_0^2(1-\kappa_0)\bigr]
\)
&
\(>0.677422123651\)
\\
pendant, level \(1\)
&
\(\displaystyle
\nu_1
=\Rfun(\kappa_1)
+\kappa_1^5
 \bigl[1-\kappa_1^2(1-\kappa_1)\bigr]
\)
&
\(>0.618537374925\)
\\
\midrule
final tree, \(|T|=2\)
&
\(\displaystyle
\Rfun(a)+\frac{a^3}{1+a^2}
\)
&
\(>0.653316253302\)
\\
final tree, \(|T|\ge3\)
&
\(\displaystyle
\Rfun(a)+\frac23a^3
\)
&
\(>0.636408627138\)
\\
\bottomrule
\end{tabular}
\end{adjustbox}
\end{table}

Among the imported coefficients, the smallest is
\(\chi_1^{(7)}\), with
\(\chi_1^{(7)}-c>1.33\times10^{-4}\).
Appendix~\ref{app:numeric} gives certified rational bounds for the radicals appearing in these evaluations, 
so the comparisons in Table~\ref{tab:coefficients} do not rely on floating-point rounding.

It remains to improve the bound for final components.  
In the final case of Thiery's Theorem~3.7.1, a component \(T\subseteq I_3\) is bounded uniformly by the terminal coefficient
\(\rho_3(\eps_2)\), which is the size-one value of the final root bound
\citep[Theorem~3.7.1]{Thiery2023Thesis}.
For \(|T|\ge2\), we can do better by retaining the positive internal-edge contribution in
Lemma~\ref{lem:final-root-slack}.  We first locate these internal edges within the layer decomposition.

\begin{lemma}\label{lem:final-edge-layer}
Let \(T\subseteq I_3\) be a final component in the high-girth decomposition.
Then \(E(H[T])=E(H_{\le3}[T])\).
In particular, every oriented edge \((u,v)\in E(H[T])\) satisfies
\(w_u>(1-\eps_3)w_v=a^2w_v\).
\end{lemma}

\begin{proof}
By the definition of \(I_3\), each of its components is a maximal connected component of
\(H_{\le3}\) of size at most \(\ell-1\).  Thus, for a final component \(T\subseteq I_3\),
the graph \(H_{\le3}[T]\) is connected and spans \(T\).

After the Arkin-Hassin high-girth reduction, 
Corollary~3.2.8 of Thiery implies that the full induced exchange graph \(H[T]\) is a tree
\citep[Corollary~3.2.8]{Thiery2023Thesis}.
Since \(H_{\le3}[T]\) is a connected spanning subgraph of this tree, it must contain all \(|T|-1\) edges of \(H[T]\).  
Therefore \(H_{\le3}[T]=H[T]\).

Every edge of \(H[T]\) consequently belongs to one of \(H_1,H_2,H_3\).
If \((u,v)\in H_i\) for some \(i\le3\), the definition of the exchange layers gives \(w_u/w_v>1-\eps_i\ge1-\eps_3\).
Since \(1-\eps_3=a^2\), the claimed weight relation follows.
\end{proof}

\begin{lemma}\label{lem:final-subtree-local}
Let \(T\subseteq I_3\) be a final component in the high-girth decomposition with \(|T|\ge2\).
Then every connected subtree \(R\) of \(H[T]\) satisfies \(w^2(C_R)\le w^2(N_R^+).\)
Hence \(T\) satisfies the local-optimality hypothesis of Lemma~\ref{lem:final-root-slack}.
\end{lemma}

\begin{proof}
Fix a connected subtree \(R\) of \(H[T]\).
If \(C_R=\varnothing\), the claim is immediate, so assume that \(C_R\ne\varnothing\).
Because the claws \(C_v\), \(v\in A\), partition the optimum packing \(O\), 
the family \(C_R=\bigcup_{v\in R}C_v\) is independent.  
Moreover, every claw has at most three talons and every final component has size at most \(\ell-1\).  
Therefore
\[
    |C_R|
    \le 3|R|
    \le 3(\ell-1)
    =18
    <3s_0.
\]

No member of \(C_R\) belongs to the current packing \(A\).
Indeed, if \(o\in A\cap O\), then \(o\) is disjoint from every other member of both packings.  
Hence \(N(o,A)=\{o\}\), the center assignment gives \(\pi(o)=o\), 
and \(o\) cannot participate in a nonsingleton exchange component.  
Thus \(C_R\subseteq\Scal\setminus A\).

The bounded square local optimality preserved by the Arkin-Hassin lift therefore gives
\(w^2(C_R)\le w^2(N(C_R,A)).\)
For every \(o\in C_v\) with \(v\in R\), 
each neighbor of \(o\) in \(A\) is either \(v\), 
which belongs to \(N_v^+\), or belongs to \(N(o,A\setminus\{v\})\), 
which is also contained in \(N_v^+\).
Consequently
\(
    N(C_R,A)\subseteq
    \bigcup_{v\in R}N_v^+
    =N_R^+.
\)
Since all weights are positive, \(w^2(N(C_R,A))\le w^2(N_R^+)\),
and the desired inequality follows.
\end{proof}

We can now combine Lemmas~\ref{lem:final-edge-layer},
\ref{lem:final-subtree-local}, and~\ref{lem:final-root-slack}.

\begin{lemma}\label{lem:final-nonsingleton}
Every final component \(T\subseteq I_3\) with \(|T|\ge2\) satisfies
\[
\sum_{v\in T}
\left(
\frac{\Delta_v}{w_v}
+
\frac{\Psi_v}{w_v}
\right)
\ge
Q(T)+(c-3\delta)w(T).
\]
\end{lemma}

\begin{proof}
Let \(t=|T|\ge2\).
By Lemma~\ref{lem:final-subtree-local}, 
every connected subtree of \(H[T]\) satisfies the local-optimality hypothesis of Lemma~\ref{lem:final-root-slack}.  
Because \(\eps_2=\delta\) and \(1-\eps_2=a\), that lemma yields
\[
\sum_{v\in T}
\left(
\frac{\Delta_v}{w_v}
+
\frac{\Psi_v}{w_v}
\right)
\ge
Q(T)
+
\bigl(\rho_3(\delta)-3\delta\bigr)w(T)
+
a\sum_{(u,v)\in E(H[T])}w_u .
\]
Thus it remains only to lower-bound the total weight of the tails of the directed edges of \(H[T]\).

Suppose first that \(t=2\), and let \((u,v)\) be the unique oriented edge.
By Lemma~\ref{lem:final-edge-layer}, \(w_u\ge a^2w_v\).
Since \(x/(1+x)\) is increasing for \(x>0\),
\[
\frac{w_u}{w(T)}
=
\frac{w_u/w_v}{1+w_u/w_v}
\ge
\frac{a^2}{1+a^2}.
\]
Hence
\[
a\sum_{(x,y)\in E(H[T])}w_x
=
aw_u
\ge
\frac{a^3}{1+a^2}\,w(T).
\]
It follows that the coefficient multiplying \(w(T)\), before the term \(-3\delta\), is at least
\[
\rho_3(\delta)+\frac{a^3}{1+a^2}
=
\Rfun(a)+\frac{a^3}{1+a^2}
>
0.653316253302
>
c.
\]
Therefore the required inequality holds when \(t=2\).

Now assume \(t\ge3\).
By Lemma~\ref{lem:final-edge-layer}, every oriented edge \((u,v)\) satisfies \(w_u\ge a^2w_v\).
We apply the weaker estimate in \citep[Lemma~2.4.7, Eq.~(2.7)]{Thiery2023Thesis},
which is valid for every \(\varepsilon\in(0,1)\).
Taking \(1-\varepsilon=a^2\) gives
\[
\sum_{(u,v)\in E(H[T])}w_u
\ge
\frac{t-1}{t}a^2w(T).
\]
We emphasize that the stronger estimate in the same lemma, 
which requires \(\varepsilon\le1/2\), is not used here.

Substituting this bound into the final-root inequality gives
\[
\sum_{v\in T}
\left(
\frac{\Delta_v}{w_v}
+
\frac{\Psi_v}{w_v}
\right)
\ge
Q(T)
+
\left(
\rho_3(\delta)
+
\frac{t-1}{t}a^3
-
3\delta
\right)w(T).
\]
Since \((t-1)/t\) is increasing and \(t\ge3\),
\[
\rho_3(\delta)+\frac{t-1}{t}a^3
\ge
\rho_3(\delta)+\frac23a^3
=
\Rfun(a)+\frac23a^3
>
0.636408627138
>
c.
\]
Thus
\[
\sum_{v\in T}
\left(
\frac{\Delta_v}{w_v}
+
\frac{\Psi_v}{w_v}
\right)
\ge
Q(T)+(c-3\delta)w(T),
\]
as claimed.
\end{proof}

Hence every final component with at least two vertices already provides more than the coefficient \(c\) required by the global accounting.  
The only remaining final components are singletons.  
This is precisely the case in which Thiery's terminal bound can meet the \(\sqrt3\) square-locality barrier, 
and it is treated separately in the next case.

\subsubsection{Final singletons}
\label{sec:singletons}

The goal of this subsection is to prove Lemma~\ref{lem:singleton-certificate}, 
which supplies the final-component bound missing from Lemma~\ref{lem:final-nonsingleton}.

\begin{lemma}\label{lem:singleton-certificate}
For \(a=0.5433\) and \(c=0.6184\), every final singleton \(T=\{v\}\subseteq I_3\) satisfies
\[
\frac{\Delta_v+\Psi_v}{w_v}+a g_v
\ge
Q(\{v\})+(c-3\delta)w_v.
\]
\end{lemma}

The right-hand side is the same component target as in Lemma~\ref{lem:final-nonsingleton}: 
\(Q(\{v\})\) is the layer credit used in the global accounting, 
and \((c-3\delta)w_v\) is the required final-layer weight contribution.  
For a nonsingleton final component, internal edges of \(H[T]\) provide the additional positive term needed to reach this target.  
A singleton has no internal edge, so Lemma~\ref{lem:singleton-certificate} instead allows 
the positive linear claw gain \(g_v\) to supply the missing contribution.  
The proof normalizes the neighborhood of \(v\), derives an exact formula for \(\Delta_v+\Psi_v\), 
and then uses square local optimality and 
the incidence restrictions of \(3\)-Set Packing to reduce the claim to a one-variable finite check.

Let \(T=\{v\}\subseteq I_3\) be a final singleton, write \(W=w_v\), and let \(m=|C_v|\le3\).  
For the moment assume \(1\le m\le3\).  For each \(o\in C_v\), set \(y_o=w_o/W\) and let \(x=\sum_{o\in C_v}y_o\).  
Define \(B_v=N(C_v,A)\setminus\{v\}\) to be the set of distinct secondary blockers of the talons in \(C_v\).  
For \(u\in B_v\), set \(z_u=w_u/W\), let \(r_u\) be the number of talons in \(C_v\) that conflict with \(u\), 
and put \(\beta=\sum_{u\in B_v}z_u\).  
Thus \(x\) is the total normalized talon weight, 
\(\beta\) is the total normalized weight of the distinct secondary blockers, 
and \(r_u\) records the multiplicity with which blocker \(u\) appears among the talon neighborhoods.  
Finally, set \(q=a^2=1-\eps_3\).
Figure~\ref{fig:final-singleton} illustrates the local notation: \(v\) is the singleton center, 
the vertices \(o_i\) are its assigned talons, and the vertices \(u_j\) are secondary blockers.

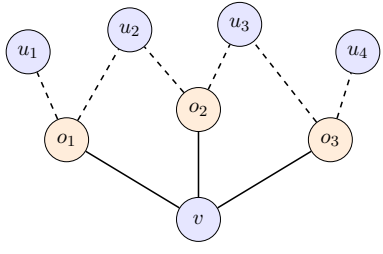
\begin{figure}[htbp]
\centering
\begin{adjustbox}{max width=.3\textwidth}
\begin{tikzpicture}[
  A/.style={circle, draw, fill=blue!10, minimum size=8mm, inner sep=1pt},
  O/.style={circle, draw, fill=orange!14, minimum size=8mm, inner sep=1pt},
  e/.style={thick},
  s/.style={thick, dashed}
]
\node[A] (v) at (0,0) {$v$};

\node[O] (o1) at (-2.4,1.45) {$o_1$};
\node[O] (o2) at (0,2.0) {$o_2$};
\node[O] (o3) at (2.4,1.45) {$o_3$};

\node[A] (b1) at (-3.1,3.05) {$u_1$};
\node[A] (b2) at (-1.25,3.45) {$u_2$};
\node[A] (b3) at (0.75,3.55) {$u_3$};
\node[A] (b4) at (2.9,3.05) {$u_4$};

\draw[e] (v)--(o1);
\draw[e] (v)--(o2);
\draw[e] (v)--(o3);

\draw[s] (o1)--(b1);
\draw[s] (o1)--(b2);
\draw[s] (o2)--(b2);
\draw[s] (o2)--(b3);
\draw[s] (o3)--(b3);
\draw[s] (o3)--(b4);
\end{tikzpicture}
\end{adjustbox}
\caption{A schematic final singleton with three talons.  
Blue vertices belong to the current packing \(A\), 
and orange vertices belong to the optimum packing \(O\).  
Solid edges indicate center-talon conflicts, while dashed edges indicate secondary-blocker incidences.  
Secondary blockers may be shared by different talons.}
\label{fig:final-singleton}
\end{figure}

The next lemma collects the local identities and structural facts used in the proof.  

\begin{lemma}\label{lem:singleton-relations}
Let \(T=\{v\}\subseteq I_3\) be a final singleton with \(1\le m=|C_v|\le3\), and use the normalized notation above.  
Then:
\begin{enumerate}[leftmargin=*,label=(\alph*)]
\item
\(
\frac{\Delta_v+\Psi_v}{W^2}
=
m+1-2x
+
\sum_{u\in B_v}
\bigl(r_uz_u-(r_u-1)z_u^2\bigr).
\)

\item Writing \(Q_v=Q(\{v\})\),
\(\frac{Q_v}{W} \le \delta\left(m+\sum_{u\in B_v}r_uz_u\right)\) and  \(\frac{g_v}{W}=\pos{x-1-\beta}.\)

\item Every talon has at most two secondary blockers.  
Consequently, \(\sum_{u\in B_v}r_u\le2m\) and \(|B_v|\le2m\).  
Moreover, every \(u\in B_v\) satisfies \(0<z_u\le q=a^2<a\).

\item
\(\sum_{u\in B_v}z_u^2
\ge
S_m(x)
:=
\pos{\frac{x^2}{m}-1}.\)
In particular, every feasible final singleton satisfies \(0\le S_m(x)\le2mq^2\).
\end{enumerate}
\end{lemma}

\begin{proof}
For part~(a), recall that
\(
N_v^+
=
\{v\}\cup\bigcup_{o\in C_v}N(o,A\setminus\{v\}).
\)
Since every \(o\in C_v\) is assigned to \(v\), we have \(v\in N(o,A)\).  
Moreover, \(B_v=N(C_v,A)\setminus\{v\}\) consists exactly of the distinct vertices of \(A\setminus\{v\}\) 
that conflict with at least one talon in \(C_v\).  Hence
\(
N(C_v,A)=N_v^+=\{v\}\mathbin{\dot\cup}B_v.
\)
Using \(w_v=W\), \(w_o=Wy_o\), and \(w_u=Wz_u\), the definition of \(\Delta_v\) therefore gives
\[
\begin{aligned}
\frac{\Delta_v}{W^2}
&=
\frac{w^2(N_v^+)-w^2(C_v)}{W^2}
=
\frac{W^2+\sum_{u\in B_v}w_u^2-\sum_{o\in C_v}w_o^2}{W^2}
=
1+\sum_{u\in B_v}z_u^2-\sum_{o\in C_v}y_o^2.
\end{aligned}
\]

We next expand \(\Psi_v\).  For a fixed \(o\in C_v\), additivity of \(w\) and \(w^2\) over vertex sets gives
\[
\begin{aligned}
&\frac{(w_o-W)^2
+W\,w(N(o,A\setminus\{v\}))
-w^2(N(o,A\setminus\{v\}))}{W^2}
=
(y_o-1)^2
+
\sum_{u\in N(o,A\setminus\{v\})}(z_u-z_u^2).
\end{aligned}
\]
Consequently,
\[
\begin{aligned}
\frac{\Psi_v}{W^2}
&=
\sum_{o\in C_v}(y_o-1)^2
+
\sum_{o\in C_v}
\sum_{u\in N(o,A\setminus\{v\})}(z_u-z_u^2)\\
&=
\sum_{o\in C_v}(y_o^2-2y_o+1)
+
\sum_{u\in B_v}r_u(z_u-z_u^2)\\
&=
\sum_{o\in C_v}y_o^2-2x+m
+
\sum_{u\in B_v}r_uz_u
-
\sum_{u\in B_v}r_uz_u^2.
\end{aligned}
\]
The second equality uses the definition of \(r_u\): 
for a fixed blocker \(u\in B_v\), the term \(z_u-z_u^2\) occurs once for each talon \(o\in C_v\) adjacent to \(u\), 
hence exactly \(r_u\) times.  
Adding the preceding expressions for \(\Delta_v/W^2\) and \(\Psi_v/W^2\), 
the terms \(\sum_{o\in C_v}y_o^2\) cancel, and we obtain
\[
\begin{aligned}
\frac{\Delta_v+\Psi_v}{W^2}
&=
m+1-2x
+
\sum_{u\in B_v}
\bigl(z_u^2+r_uz_u-r_uz_u^2\bigr)\\
&=
m+1-2x
+
\sum_{u\in B_v}
\bigl(r_uz_u-(r_u-1)z_u^2\bigr),
\end{aligned}
\]
as claimed.

For part~(b), recall that
\(
Q_v
=
\sum_{o\in C_v}\sum_{i=1}^3
\eps_{i-1}w\bigl(N(o,S_{i+1})\bigr),
\)
where \(A=S_1\mathbin{\dot\cup}S_2\mathbin{\dot\cup}S_3\mathbin{\dot\cup}S_4\) and \(S_4=I_3\).  
Since \(v\in S_4\) and every \(o\in C_v\) is assigned to \(v\), 
the vertex \(v\) belongs to \(N(o,S_4)\) for every \(o\in C_v\).  
Thus the center contributes exactly \(\eps_2W=\delta W\) for each of the \(m\) talons.

A secondary blocker \(u\in B_v\) contributes to \(Q_v\) only if \(u\in S_2\cup S_3\cup S_4\).  
If \(u\in S_j\), where \(j\in\{2,3,4\}\), then every incidence of \(u\) with a talon contributes with coefficient \(\eps_{j-2}\).  
Since \(u\) is incident to exactly \(r_u\) talons, and since \(\eps_0\le\eps_1\le\eps_2=\delta\), we obtain
\[
\begin{aligned}
\frac{Q_v}{W}
&=
m\eps_2
+\eps_0\sum_{u\in B_v\cap S_2}r_uz_u
+\eps_1\sum_{u\in B_v\cap S_3}r_uz_u
+\eps_2\sum_{u\in B_v\cap S_4}r_uz_u\\
&\le
\delta m
+\delta\sum_{u\in B_v\cap(S_2\cup S_3\cup S_4)}r_uz_u\\
&\le
\delta\left(m+\sum_{u\in B_v}r_uz_u\right).
\end{aligned}
\]
Notice that blockers in \(S_1\) do not appear in the definition of \(Q_v\), 
since the sum defining \(Q_v\) contains only \(S_2,S_3,S_4\).

Furthermore,
\(N(C_v,A)=\{v\}\mathbin{\dot\cup}B_v\), so
\(w(C_v)=Wx\) and \(w(N(C_v,A))=W(1+\beta)\).  Since \(W>0\),
\[
\begin{aligned}
\frac{g_v}{W}
&=
\frac{\pos{w(C_v)-w(N(C_v,A))}}{W}
=
\frac{\pos{Wx-W(1+\beta)}}{W}
=
\pos{x-1-\beta}.
\end{aligned}
\]
This proves part~(b).

For part~(c), fix \(o\in C_v\).  If \(o=v\in A\cap O\), 
then \(v\) is disjoint from every other member of both packings \(A\) and \(O\).  
Hence \(N(v,A)=\{v\}\); in fact \(C_v=\{v\}\), and this talon has no secondary blocker.

Suppose instead that \(o\in O\setminus A\).  
Because \(o\) contains at most three ground elements and the members of \(A\) are pairwise disjoint, 
\(o\) can intersect at most three members of \(A\): 
distinct members of \(A\) intersecting \(o\) must use distinct ground elements of \(o\).  
In the high-girth instance used by the component analysis, 
Lemma~\ref{lem:crossing-gain}(c) preserves this \(A\)-neighborhood degree under the covering projection.  
Since \(v=\pi(o)\) is one of the neighbors of \(o\), 
there are at most two additional neighbors in \(A\setminus\{v\}\), that is, at most two secondary blockers.

The quantity \(\sum_{u\in B_v}r_u\) counts precisely the talon--secondary-blocker incidences.  
Summing the preceding bound over the \(m\) talons therefore yields
\(
\sum_{u\in B_v}r_u\le2m.
\)
Since every \(u\in B_v\) is incident to at least one talon, \(r_u\ge1\), and hence
\(
|B_v|
\le
\sum_{u\in B_v}r_u
\le2m.
\)

It remains to bound the weight of each secondary blocker.  
Fix \(u\in B_v\), and choose a talon \(o\in C_v\) that conflicts with \(u\).  
By the definition of the exchange graph, 
the witness \(o\in C_v\) together with the secondary blocker \(u\in N(o,A\setminus\{v\})\) induces the directed exchange arc \((u,v)\).

We claim that this arc cannot belong to \(H_{\le3}\). 
Since \(v\in I_3\subseteq I_2\) and \(T=\{v\}\) is a final singleton,
the connected component of \(H_{\le3}[I_2]\) containing \(v\) is the singleton \(\{v\}\).  
If \((u,v)\in H_{\le3}\) and \(u\in I_2\), 
then \(u\) and \(v\) would lie in the same component of \(H_{\le3}[I_2]\), 
contradicting that this component is a singleton.  
On the other hand, if \((u,v)\in H_{\le3}\) and \(u\notin I_2\), 
then the singleton component \(\{v\}\) of \(H_{\le3}[I_2]\) would be 
contained in a strictly larger connected component of the full graph \(H_{\le3}\).  
By the definition of the layer partition, \(v\) would then belong to \(P_3\), 
rather than \(I_3\), again a contradiction.

Thus \((u,v)\notin H_{\le3}\).  Since every exchange arc belongs to exactly one of \(H_1,H_2,H_3,H_4\), 
we must have \((u,v)\in H_4\).  By the definition of \(H_4\),
\(
0<
\frac{w_u}{w_v}
\le
1-\eps_3.
\)
With \(z_u=w_u/W\), \(W=w_v\), and \(q=1-\eps_3=a^2\), this gives
\(
0<z_u\le q=a^2<a,
\)
where the last inequality follows from \(0<a<1\).  This proves part~(c).

For part~(d), first consider the common-vertex case \(v\in A\cap O\).  
Since both \(A\) and \(O\) are packings, \(v\) is isolated in \(G[A\cup O]\).  
Thus \(C_v=\{v\}\), \(B_v=\varnothing\), \(m=1\), and \(x=1\).  Hence
\[
\sum_{u\in B_v}z_u^2=0
=
\pos{\frac{x^2}{m}-1}
=
S_m(x),
\]
so the desired lower bound holds.

Assume now that \(v\notin O\).  Then every member of \(C_v\) lies in \(O\setminus A\): 
indeed, if \(o\in A\cap O\), then \(o\) is isolated from every other member of \(A\), 
so its unique member of \(N(o,A)\) is \(o\) itself and therefore \(\pi(o)=o\), not \(v\).  
Consequently, \(C_v\subseteq\Scal\setminus A.\)
Because \(O\) is a packing, \(C_v\) is an independent family.  
Moreover, \(|C_v|=m\le3\le3s_0,\)
so \(C_v\) is an admissible \(s_0\)-exchange in the sense of Phase~I.  
The high-girth reduction preserves the bounded square local optimality used in the analysis.  
Hence the exchange inserting \(C_v\) and removing \(N(C_v,A)\) cannot strictly increase the squared-weight potential:
\(w^2(C_v)\le w^2(N(C_v,A)).\)
Using \(N(C_v,A)=\{v\}\mathbin{\dot\cup}B_v\) and dividing by \(W^2\), we obtain
\[
\begin{aligned}
\sum_{o\in C_v}y_o^2
&=
\frac{w^2(C_v)}{W^2}
\le
\frac{w^2(N(C_v,A))}{W^2}
=
1+\sum_{u\in B_v}z_u^2.
\end{aligned}
\]
Since \(m\ge1\), Cauchy-Schwarz applied to the \(m\) numbers \((y_o)_{o\in C_v}\) gives
\(
x^2
=
\left(\sum_{o\in C_v}y_o\right)^2
\le
m\sum_{o\in C_v}y_o^2,
\)
and therefore
\(
\sum_{u\in B_v}z_u^2
\ge
\sum_{o\in C_v}y_o^2-1
\ge
\frac{x^2}{m}-1.
\)
The left-hand side is nonnegative, so the two lower bounds combine to give
\[
\sum_{u\in B_v}z_u^2
\ge
\max\left\{0,\frac{x^2}{m}-1\right\}
=
S_m(x).
\]

Finally, part~(c) shows that \(|B_v|\le2m\) and \(z_u\le q\) for every \(u\in B_v\).  Hence
\[
\begin{aligned}
S_m(x)
&\le
\sum_{u\in B_v}z_u^2
\le
|B_v|q^2
\le
2mq^2.
\end{aligned}
\]
Since \(S_m(x)\ge0\) by definition, this proves
\(0\le S_m(x)\le2mq^2\) and completes part~(d).
\end{proof}

We next convert the lower bound on \(\sum z_u^2\) into a lower bound on the total normalized blocker weight \(\beta=\sum z_u\).  
For \(0\le S\le2mq^2\), let \(\beta_{\min}(S)\) denote 
the minimum possible value of \(\sum_i z_i\) over at most \(2m\) coordinates satisfying \(0\le z_i\le q\) and \(\sum_i z_i^2\ge S\).  
To write this minimum explicitly, decompose \(S=hq^2+\theta\)
with \(h=\lfloor S/q^2\rfloor\) and \(0\le\theta<q^2\).  
This decomposition reflects the extremal form of a minimizing vector: 
each coordinate contributes at most \(q^2\) to \(\sum_i z_i^2\), 
so \(h\) records the number of coordinates that can be filled to the cap \(q\), 
while \(\theta\) is the remaining required squared contribution.

\begin{lemma}\label{lem:l1-l2}
For \(0\le S\le2mq^2\), if \(S=hq^2+\theta\) with \(h\in\mathbb Z_{\ge0}\) and \(0\le\theta<q^2\), then
\[
\beta_{\min}(S)=hq+\sqrt{\theta}.
\]
\end{lemma}

\begin{proof}
The case \(S=0\) is immediate.  For \(S>0\), pad any feasible vector with zeros to obtain \(2m\) coordinates.  
At a minimizer of \(\sum_i z_i\), the constraint \(\sum_i z_i^2\ge S\) is tight; 
otherwise the positive coordinates can be decreased slightly while preserving feasibility.

Suppose two coordinates satisfy \(0<x\le y<q\).  
If \(x^2+y^2\le q^2\), replacing \((x,y)\) by \((0,\sqrt{x^2+y^2})\) preserves \(x^2+y^2\) and strictly decreases \(x+y\).  
If \(x^2+y^2>q^2\), replace \((x,y)\) by \((\sqrt{x^2+y^2-q^2},q)\).  
Again the squared sum is unchanged, and
\[
\begin{aligned}
(x+y-q)^2-(x^2+y^2-q^2)
&=2(x-q)(y-q)>0.
\end{aligned}
\]
Since \(x^2+y^2>q^2\) implies \(x+y-q>0\), the new pair also has strictly smaller sum.  
Thus a minimizer has at most one coordinate strictly between \(0\) and \(q\).  
Because its squared sum is \(S=hq^2+\theta\), it consists of \(h\) coordinates equal to \(q\), 
one coordinate equal to \(\sqrt\theta\) when \(\theta>0\), and zeros otherwise.  
Hence \(\beta_{\min}(S)=hq+\sqrt\theta\).
\end{proof}

We are now ready to prove Lemma~\ref{lem:singleton-certificate}.  
The preceding lemmas reduce the blocker structure to the aggregate parameters \(x\) and \(\beta\), 
and then lower-bound \(\beta\) in terms of \(x\).  
This reduces the desired inequality to a one-variable verification for \(m=1,2,3\).

\begin{proof}[Proof of Lemma~\ref{lem:singleton-certificate}]
If \(m=0\), then \(C_v=\varnothing\), \(N_v^+=\{v\}\), and therefore \(\Delta_v=W^2\), while \(\Psi_v=Q(\{v\})=g_v=0\).  
The claim reduces to \(W\ge(c-3\delta)W\), which is immediate.  
Assume henceforth that \(1\le m\le3\), and write \(Q_v=Q(\{v\})\).

Using Lemma~\ref{lem:singleton-relations}(a)-(c), \(\delta=1-a\), and \(z_u<a\), we obtain
\begin{equation}
\begin{aligned}
&\frac{\Delta_v+\Psi_v}{W^2}
+a\frac{g_v}{W}
-\frac{Q_v}{W}
+3\delta\\
&\quad\ge
m+1-2x+(3-m)(1-a)
+
\sum_{u\in B_v}
\bigl(ar_uz_u-(r_u-1)z_u^2\bigr)
+a\pos{x-1-\beta}\\
&\quad=
m+1-2x+(3-m)(1-a)
+a\beta+a\pos{x-1-\beta}
+
\sum_{u\in B_v}(r_u-1)z_u(a-z_u)\\
&\quad\ge
m+1-2x+(3-m)(1-a)
+a\max\{\beta,x-1\}
\quad=:
F_m(x,\beta). \label{eq:Fmbound}
\end{aligned}
\end{equation}
Thus \(F_m(x,\beta)\) is a lower bound on the normalized left-hand side needed for the singleton certificate 
after the layer credit has been subtracted.  
The only remaining dependence on the blocker vector is through its total normalized weight \(\beta\).

Lemma~\ref{lem:singleton-relations}(d) gives the required lower bound \(S_m(x)\) on \(\sum z_u^2\), 
while Lemma~\ref{lem:l1-l2} converts that requirement into a lower bound on \(\beta\).  
By its definition, \(\beta_{\min}\) is nondecreasing, and therefore
\[
\begin{aligned}
\beta
&\ge
\beta_{\min}\!\left(\sum_{u\in B_v}z_u^2\right)
\ge
\beta_{\min}(S_m(x)).
\end{aligned}
\]
Because \(F_m(x,\beta)\) is nondecreasing in \(\beta\), define
\[
\Phi_m(x)
:=
F_m\bigl(x,\beta_{\min}(S_m(x))\bigr)
=
m+1-2x+(3-m)(1-a)
+a\max\left\{\beta_{\min}(S_m(x)),x-1\right\}.
\]
Every feasible singleton therefore satisfies \(F_m(x,\beta)\ge\Phi_m(x)\).  
The reduction follows the chain
\(
x
\longmapsto
S_m(x)
\longmapsto
\beta_{\min}(S_m(x))
\longmapsto
\Phi_m(x):
\)
the talon weight \(x\) forces a minimum value of \(\sum z_u^2\), 
that requirement forces a minimum total blocker weight, and substituting the latter into \(F_m\) leaves a function of \(x\) alone.  
By Lemma~\ref{lem:singleton-relations}(d), its necessary domain is \(0\le S_m(x)\le2mq^2\).

It remains to minimize \(\Phi_m\) on this domain.  
When \(S_m(x)=0\), \(\beta_{\min}(S_m(x))=0\), and \(\Phi_m\) is piecewise linear with slopes \(-2\) and \(-2+a\), both negative.  
Hence its minimum on this region occurs at \(x=\sqrt m\).

Now suppose \(S_m(x)>0\).  For \(h\in\{0,\ldots,2m-1\}\), on an interval where \(hq^2\le S_m(x)\le(h+1)q^2\), 
Lemma~\ref{lem:l1-l2} gives
\[
\beta_{\min}(S_m(x))
=
hq+
\sqrt{\frac{x^2}{m}-1-hq^2}.
\]
The integer \(h\) changes only when \(S_m(x)=hq^2\), equivalently at
\(
x_{m,h}=\sqrt{m(1+hq^2)}
\)
for \(h=0,\ldots,2m\); the point \(h=2m\) is the final feasible endpoint. 
Within each interval, the square-root term is concave in \(x\).  
The active term in the maximum defining \(\Phi_m\) can change only at a solution of \(\beta_{\min}(S_m(x))=x-1\).  
After splitting at these crossing points, \(\Phi_m\) is either linear or concave, 
so its minimum on each piece is attained at an endpoint.  
Therefore it is enough to check the points \(x_{m,h}\) together with the admissible crossings.  
Appendix~\ref{app:finite-checks} lists all such candidates and verifies the required branch and sign conditions.

For \(m=1\), the minimum occurs at \(x=1\), where
\(
\Phi_1(1)=2(1-a)=0.9134.
\)
For \(m=2\), the minimum occurs at \(S_2(x)=2q^2\), namely at \(x=\sqrt{2(1+2q^2)}=\sqrt{2(1+2a^4)}\), and
\[
\Phi_2(x)
=
4-a-2\sqrt{2(1+2a^4)}+2a^3
>
0.712465347026.
\]
For \(m=3\), the minimum occurs at the admissible crossing on the branch \(3q^2\le S_3(x)\le4q^2\), where
\(
3q+\sqrt{x^2/3-1-3q^2}=x-1.
\)
The smaller root is
\(
x_*=[3(1+3q)-\sqrt{9q^2+18q-3}]/2,
\)
with \(1.9483293410782<x_*<1.9483293410783\).  At this crossing,
\[
\Phi_3(x_*)
=
4-a-(2-a)x_*
>
0.6185686488
>
c.
\]
Consequently \(\Phi_m(x)\ge c\) for every feasible \(x\) and every \(m\in\{1,2,3\}\).  
Combining this with Eq.~\eqref{eq:Fmbound} and multiplying by \(W=w_v\) gives
\[
\frac{\Delta_v+\Psi_v}{w_v}
+a g_v
-Q_v
+3\delta w_v
\ge
c w_v,
\]
which is equivalent to the statement of Lemma~\ref{lem:singleton-certificate}.
\end{proof}

\subsubsection{The augmented structural inequality}
\label{sec:structural}

We now combine the component bounds on the high-girth lift and transfer the resulting inequality back to the original instance.  
The only additional term is the linear claw gain needed for final singleton components.

\begin{theorem}\label{thm:augmented-structural}
For the original Phase~I packing \(A\) and every optimum packing \(O\),
\[
2w(O)
\le
(4-c)w(A)+a\Gcal(A,O)
=
3.3816\,w(A)+0.5433\,\Gcal(A,O).
\]
\end{theorem}

\begin{proof}
Apply the Arkin-Hassin high-girth reduction with the target girth required by Thiery's layered decomposition.  
Let \((\widetilde G,\widetilde A,\widetilde O)\)
be the resulting \(M\)-sheet covering, equipped with the lifted center assignment of Lemma~\ref{lem:crossing-gain}.  
Let \(\widetilde S_1,\widetilde S_2,\widetilde S_3,\widetilde S_4\)
be the corresponding layer classes, with \(\widetilde S_4=\widetilde I_3\), 
and let \(\widetilde J\subseteq\widetilde S_4\)
be the set of centers whose final component is a singleton.

Define the total layer credit
\(
\widetilde Q
:=
\sum_{\widetilde o\in\widetilde O}
\sum_{i=1}^{3}
\eps_{i-1}\,
w\bigl(N(\widetilde o,\widetilde S_{i+1})\bigr).
\)
The decomposition trees partition \(\widetilde A\), and the claws
\(\{\widetilde C_{\widetilde v}:\widetilde v\in\widetilde A\}\)
partition \(\widetilde O\).  Hence the \(Q(T)\)-terms from all lifted decomposition trees sum exactly to \(\widetilde Q\):
\[
\begin{aligned}
\sum_T Q(T)
&=
\sum_T
\sum_{\widetilde v\in T}
\sum_{\widetilde o\in\widetilde C_{\widetilde v}}
\sum_{i=1}^{3}
\eps_{i-1}\,
w\bigl(N(\widetilde o,\widetilde S_{i+1})\bigr)
=
\sum_{\widetilde o\in\widetilde O}
\sum_{i=1}^{3}
\eps_{i-1}\,
w\bigl(N(\widetilde o,\widetilde S_{i+1})\bigr)
=
\widetilde Q.
\end{aligned}
\]

We next sum the bounds over the decomposition components.  
For the first layer, \(\widetilde S_1=\widetilde D_1\), and the corresponding root coefficient is \(\chi_{-1}^{(7)}\).  
For \(i\in\{1,2\}\), the vertices of \(\widetilde S_{i+1}=\widetilde D_{i+1}\mathbin{\dot\cup}\widetilde P_{i+1}\) 
occur in root and pendant pieces with coefficients \(\chi_{i-1}^{(7)}\) and \(\nu_{i-1}\), respectively.  
Table~\ref{tab:coefficients} verifies that all these coefficients, 
\(\chi_{-1}^{(7)},\ \chi_0^{(7)},\ \chi_1^{(7)},\ \nu_0,\ \nu_1\), are at least \(c\).
Thiery's root and pendant bounds therefore imply a contribution of at least \(c\,w(\widetilde S_1)\) from the first layer and, 
since \(k=3\), at least \(\bigl(c-3\eps_{i-1}\bigr)w(\widetilde S_{i+1})\)
from each layer \(\widetilde S_{i+1}\), \(i\in\{1,2\}\), 
in addition to the corresponding layer-credit terms \citep[Lemmas~3.5.1, 3.5.2, and~3.6.1, and Eq.~(3.11)]{Thiery2023Thesis}.

For the final class \(\widetilde S_4=\widetilde I_3\),
Lemma~\ref{lem:final-nonsingleton} gives the coefficient \(c-3\eps_2=c-3\delta\)for every nonsingleton component.  
Lemma~\ref{lem:singleton-certificate} gives the same coefficient for a singleton \(\{\widetilde v\}\)
after adding \(a\widetilde g_{\widetilde v}\).
Summing over all components therefore gives
\[
\begin{aligned}
&\sum_{\widetilde v\in\widetilde A}
\left(
\frac{\widetilde\Delta_{\widetilde v}}{w_{\widetilde v}}
+
\frac{\widetilde\Psi_{\widetilde v}}{w_{\widetilde v}}
\right)
+
a\sum_{\widetilde v\in\widetilde J}
\widetilde g_{\widetilde v}
\ge
\widetilde Q
+
c\,w(\widetilde S_1)
+
\sum_{i=1}^{3}
\bigl(c-3\eps_{i-1}\bigr)
w(\widetilde S_{i+1}).
\end{aligned}
\]

For \(\widetilde v\in\widetilde S_j\), let \(d_{\widetilde O}(\widetilde v)\)
be the number of \(\widetilde o\in\widetilde O\)
such that \(\widetilde v\in N(\widetilde o,\widetilde A)\).
If \(\widetilde v\notin\widetilde O\), local bijectivity of the covering preserves its incidence degree with the optimum packing, 
which is at most \(3\): a set of cardinality at most \(3\) can intersect at most three pairwise disjoint members of \(O\).  
If \(\widetilde v\in\widetilde A\cap\widetilde O\), 
then no other member of \(\widetilde O\) conflicts with \(\widetilde v\), 
so \(d_{\widetilde O}(\widetilde v)=1\).
Consequently,
\[
\begin{aligned}
\sum_{\widetilde o\in\widetilde O}
w\bigl(N(\widetilde o,\widetilde S_j)\bigr)
&=
\sum_{\widetilde v\in\widetilde S_j}
d_{\widetilde O}(\widetilde v)\,w_{\widetilde v}
\le
3w(\widetilde S_j).
\end{aligned}
\]

We can now apply the master inequality Eq.~\eqref{eq:master}.  
Since \(\widetilde S_1,\ldots,\widetilde S_4\)
partition \(\widetilde A\), substituting the component bound above, 
cancelling the layer credit \(\widetilde Q\), and then applying the neighborhood bound gives
\begin{align*}
2w(\widetilde O)
\le\;&
w(\widetilde A)
+
\sum_{j=1}^{4}
\sum_{\widetilde o\in\widetilde O}
w\bigl(N(\widetilde o,\widetilde S_j)\bigr)
-
\sum_{\widetilde v\in\widetilde A}
\left(
\frac{\widetilde\Delta_{\widetilde v}}{w_{\widetilde v}}
+
\frac{\widetilde\Psi_{\widetilde v}}{w_{\widetilde v}}
\right)
\\
\le\;&
w(\widetilde A)
+
\sum_{j=1}^{4}
\sum_{\widetilde o\in\widetilde O}
w\bigl(N(\widetilde o,\widetilde S_j)\bigr)
% \\
% &\quad
-
\widetilde Q
-
c\,w(\widetilde S_1)
-
\sum_{i=1}^{3}
\bigl(c-3\eps_{i-1}\bigr)
w(\widetilde S_{i+1})
+
a\sum_{\widetilde v\in\widetilde J}
\widetilde g_{\widetilde v}
\\
=\;&
w(\widetilde A)
+
\sum_{\widetilde o\in\widetilde O}
w\bigl(N(\widetilde o,\widetilde S_1)\bigr)
-
c\,w(\widetilde S_1)
% \\
% &\quad
+
\sum_{i=1}^{3}
\left[
\sum_{\widetilde o\in\widetilde O}
w\bigl(N(\widetilde o,\widetilde S_{i+1})\bigr)
-
\eps_{i-1}
\sum_{\widetilde o\in\widetilde O}
w\bigl(N(\widetilde o,\widetilde S_{i+1})\bigr)
\right]
\\
&\quad
-
\sum_{i=1}^{3}
\bigl(c-3\eps_{i-1}\bigr)
w(\widetilde S_{i+1})
+
a\sum_{\widetilde v\in\widetilde J}
\widetilde g_{\widetilde v}
\\
=\;&
w(\widetilde A)
+
\sum_{\widetilde o\in\widetilde O}
w\bigl(N(\widetilde o,\widetilde S_1)\bigr)
-
c\,w(\widetilde S_1)
% \\
% &\quad
+
\sum_{i=1}^{3}
(1-\eps_{i-1})
\sum_{\widetilde o\in\widetilde O}
w\bigl(N(\widetilde o,\widetilde S_{i+1})\bigr)
\\
&\quad
-
\sum_{i=1}^{3}
\bigl(c-3\eps_{i-1}\bigr)
w(\widetilde S_{i+1})
+
a\sum_{\widetilde v\in\widetilde J}
\widetilde g_{\widetilde v}
\\
\le\;&
w(\widetilde A)
+
3w(\widetilde S_1)
-
c\,w(\widetilde S_1)
% \\
% &\quad
+
\sum_{i=1}^{3}
\left[
3(1-\eps_{i-1})
-
\bigl(c-3\eps_{i-1}\bigr)
\right]
w(\widetilde S_{i+1})
+
a\sum_{\widetilde v\in\widetilde J}
\widetilde g_{\widetilde v}
\\
=\;&
w(\widetilde A)
+
(3-c)w(\widetilde S_1)
+
\sum_{i=1}^{3}
(3-c)w(\widetilde S_{i+1})
+
a\sum_{\widetilde v\in\widetilde J}
\widetilde g_{\widetilde v}
\\
=\;&
w(\widetilde A)
+
(3-c)
\sum_{j=1}^{4}
w(\widetilde S_j)
+
a\sum_{\widetilde v\in\widetilde J}
\widetilde g_{\widetilde v}
\\
=\;&
(4-c)w(\widetilde A)
+
a\sum_{\widetilde v\in\widetilde J}
\widetilde g_{\widetilde v}.
\end{align*}

The second inequality substitutes the summed component bound, written as an upper bound on the negative slack term.  
The next equality expands \(\widetilde Q\), 
so that for each \(i\in\{1,2,3\}\) the layer credit cancels 
an \(\eps_{i-1}\)-fraction of the neighborhood contribution of \(\widetilde S_{i+1}\); 
there is no corresponding term for \(\widetilde S_1\).  
The following inequality uses \(\sum_{\widetilde o\in\widetilde O}w(N(\widetilde o,\widetilde S_j))\le3w(\widetilde S_j)\) for every \(j\).  Finally, \(3(1-\eps_{i-1})-(c-3\eps_{i-1})=3-c\), so every layer receives the same residual coefficient \(3-c\).  
Since the four layer classes partition \(\widetilde A\), 
their total weight is \(w(\widetilde A)\), which yields the coefficient \(4-c\).

It remains to transfer this inequality back to the original instance.  
Since the covering has \(M\) sheets and preserves vertex weights,
\(w(\widetilde A)=Mw(A)\) and \(w(\widetilde O)=Mw(O)\).
Moreover, Lemma~\ref{lem:crossing-gain} applied to \(\widetilde J\subseteq\widetilde A\)
gives 
\(
\sum_{\widetilde v\in\widetilde J}
\widetilde g_{\widetilde v}
\le
M\Gcal(A,O).
\)
Therefore
\[
\begin{aligned}
2Mw(O)
=
2w(\widetilde O)
&\le
(4-c)w(\widetilde A)
+
a\sum_{\widetilde v\in\widetilde J}
\widetilde g_{\widetilde v}
\le
(4-c)Mw(A)+aM\Gcal(A,O),
\end{aligned}
\]
and division by \(M\) yields \(2w(O)\le(4-c)w(A)+a\Gcal(A,O).\)
Finally, \(c=0.6184\) and \(a=0.5433\), so \(4-c=3.3816\), proving the stated numerical form.
\end{proof}

% ===============================================
\subsection{Analysis of Phase II}
\label{sec:repair}
% ===============================================

The residual term \(\Gcal(A,O)\) in Theorem~\ref{thm:augmented-structural} depends on the unknown optimum \(O\) 
and is therefore unavailable to the algorithm.  
Its definition, however, involves only claws \(C_v\) with at most three talons.  
Phase~II exploits this fact indirectly: 
it considers every independent \(X\subseteq\Scal\setminus A\) with \(1\le |X|\le3\) and \(\gain_A(X)>0\).  
The positive-gain optimum claws form a feasible, but unknown, 
solution to the resulting auxiliary weighted \(9\)-Set Packing instance.

\begin{lemma}\label{lem:repair-benchmark}
The optimum value of the auxiliary weighted \(9\)-Set Packing instance constructed in Phase~II is at least \(\Gcal(A,O).\)
\end{lemma}

\begin{proof}
Fix \(v\in A\) with \(g_v>0\).  As observed earlier for common vertices, 
if \(v\in A\cap O\), then \(C_v=\{v\}\) and \(g_v=0\).  Hence \(v\notin O\).

In fact, \(C_v\cap A=\varnothing\).  If \(o\in C_v\cap A\), then \(o\in A\cap O\).  
Its inclusive neighborhood is \(N(o,A)=\{o\}\), so the center assignment necessarily gives \(\pi(o)=o\).  
Since \(o\in C_v\) also means \(\pi(o)=v\), we would have \(o=v\), contradicting \(v\notin O\).  
Thus \(C_v\subseteq\Scal\setminus A\).

The claw \(C_v\) is a subfamily of the packing \(O\), so it is independent, and. \(|C_v|\le3\).  Moreover,
\[
\gain_A(C_v)
=
w(C_v)-w(N(C_v,A))
=
g_v
>
0.
\]
Therefore \(C_v\) is among the Phase~II repair candidates before candidates with identical ground unions are merged.  
Its ground union \(\Gamma(C_v)=\bigcup_{F\in C_v}F\) has cardinality at most \(9\).  
If another candidate has the same ground union and larger declared gain, 
Phase~II retains that candidate instead; consequently the auxiliary weight assigned to \(\Gamma(C_v)\) is at least \(g_v\).

The claws \(\{C_v:v\in A\}\) partition \(O\).  Since \(O\) is a packing, the ground unions
\(\Gamma(C_v)\) corresponding to distinct positive-gain claws are pairwise disjoint.  
They therefore form a feasible auxiliary packing of total weight at least
\(
\sum_{v:g_v>0}g_v
=
\sum_{v\in A}g_v
=
\Gcal(A,O).
\)
Hence the auxiliary optimum is at least \(\Gcal(A,O)\).
\end{proof}

Let \(\mathcal X\) denote the repair candidates represented by the auxiliary packing returned in Phase~II, 
and define their total declared gain by
\(
\widehat\Delta
:=
\sum_{X\in\mathcal X}\gain_A(X).
\)
Recall that
\(
A^+
=
\left(
A\setminus\bigcup_{X\in\mathcal X}N(X,A)
\right)
\cup
\bigcup_{X\in\mathcal X}X,
\)
and let
\(\Delta_{\rm rep}:=w(A^+)-w(A)\)
be the actual improvement in the original objective.

\begin{lemma}\label{lem:gain-conservative}
The actual repair gain satisfies
\(
\Delta_{\rm rep}\ge\widehat\Delta.
\)
Consequently,
\(
\Gcal(A,O)\le5.5\,\Delta_{\rm rep}.
\)
\end{lemma}

\begin{proof}
Because the selected auxiliary sets form a packing, the ground unions
\(\Gamma(X)\), \(X\in\mathcal X\), are pairwise disjoint.  Hence all entering sets in
\(\bigcup_{X\in\mathcal X}X\) are mutually disjoint, and the selected repairs can be applied simultaneously.  
By the definition of \(A^+\),
\[
\Delta_{\rm rep}
=
\sum_{X\in\mathcal X}w(X)
-
w\!\left(
\bigcup_{X\in\mathcal X}N(X,A)
\right).
\]
The blocker neighborhoods may overlap.  Since \(w\) is additive on families of sets,
\(
w\!\left(
\bigcup_{X\in\mathcal X}N(X,A)
\right)
\le
\sum_{X\in\mathcal X}w(N(X,A)).
\)
It follows that
\(
\Delta_{\rm rep}
\ge
\sum_{X\in\mathcal X}
\bigl[w(X)-w(N(X,A))\bigr]
=
\widehat\Delta.
\)

It remains to lower-bound \(\widehat\Delta\).  
Every auxiliary set has cardinality at most \(9\), so the conflict graph of the auxiliary instance is \(10\)-claw free.  
Berman's squared-weight local search therefore gives a \(5\)-approximation for an exact local optimum \citep{Berman2000}.  
The polynomial-time implementation fixed in Phase~II uses the standard weight-scaling argument with \(\eta=1/11\).  
Under this scaling, an approximation factor \(\alpha\) becomes \(\alpha/(1-\eta)\) \citep[Section~2.5.3]{Thiery2023Thesis}.  
Thus the auxiliary approximation factor is
\(
\frac{5}{1-1/11}
=
\frac{11}{2}
=
5.5.
\)

Let \(\mathrm{OPT}_{\rm aux}\) denote the optimum auxiliary value.  
The Phase~II guarantee and Lemma~\ref{lem:repair-benchmark} give
\(
\widehat\Delta
\ge
\frac{\mathrm{OPT}_{\rm aux}}{5.5}
\ge
\frac{\Gcal(A,O)}{5.5}.
\)
Combining this with \(\Delta_{\rm rep}\ge\widehat\Delta\) yields
\(
\Gcal(A,O)
\le
5.5\,\Delta_{\rm rep},
\)
as claimed.
\end{proof}

% ===============================================
\subsection{Completion of the proof and polynomial-time implementation}
\label{sec:poly-implementation}
% ===============================================

We first combine the Phase~I structural inequality with the Phase~II repair guarantee.  
We then show that the conceptual two-phase procedure 
admits a polynomial-time implementation for arbitrary positive rational weights 
without changing the approximation factor.

\begin{proof}[Proof for the approximation guarantee in Theorem~\ref{thm:main}]
By Theorem~\ref{thm:augmented-structural} and Lemma~\ref{lem:gain-conservative},
\[
2w(O)
\le
(4-c)w(A)+a\Gcal(A,O)
\le
(4-c)w(A)+5.5a\,\Delta_{\rm rep}.
\]
Since \(5.5a=5.5\cdot0.5433=2.98815<3.3816=4-c\),
\[
2w(O)
\le
(4-c)\bigl(w(A)+\Delta_{\rm rep}\bigr)
=
(4-c)w(A^+)
=
3.3816\,w(A^+).
\]
Thus \(w(O)\le1.6908\,w(A^+)\).

The absorption margin is
\((4-c)-5.5a=3.3816-2.98815=0.39345\).
Hence the auxiliary approximation used in Phase~II is not the limiting numerical constraint; 
the final constant is determined by the Phase~I structural bounds.
\end{proof}

The Phase~I analysis of Algorithm~\ref{alg:two phase} assumes exact square local optimality.  
For arbitrary rational weights, a sequence of strict square improvements need not have polynomial length. 
 We therefore apply the standard scaling and one-round partial-enumeration argument 
 for weighted squared local search \citep[Section~2.5.3]{Thiery2023Thesis}.

Let \(n=|\Scal|\), \(\alpha=1.6908\), and \(\zeta=(\alpha-1)/n\).  
The empty instance is trivial.  For each \(F\in\Scal\), 
let \(\Scal_F=\{H\in\Scal:H\cap F=\varnothing\}\).
If \(\Scal_F\neq\varnothing\), compute its standard greedy packing \(P_0\) in nonincreasing order of the original weights and set
\[
\lambda=\frac{n}{\zeta w(P_0)},
\qquad
\bar w(H)=\lfloor\lambda w(H)\rfloor,
\quad H\in\Scal_F.
\]
Sets with \(\bar w(H)=0\) are discarded, and the two-phase core is run on the remaining family using the integral weights \(\bar w\).

\begin{algorithm}[t]
\caption{Polynomial wrapper for the two-phase core}
\label{alg:poly-wrapper}
\centering
\begin{adjustbox}{max width=\textwidth}
\begin{minipage}{0.98\textwidth}
\begin{algorithmic}
\Statex \textbf{Input:} weighted family \((\Scal,w)\) with positive rational weights
\Statex \textbf{Output:} a packing of \(\Scal\)

\If{\(\Scal=\varnothing\)}
    \State \Return \(\varnothing\).
\EndIf

\State \(n\gets|\Scal|\), \(\alpha\gets1.6908\), \(\zeta\gets(\alpha-1)/n\), and \(\mathcal B\gets\varnothing\).

\For{each \(F\in\Scal\)}
    \State \(\Scal_F\gets\{H\in\Scal:H\cap F=\varnothing\}\).

    \If{\(\Scal_F=\varnothing\)}
        \State Add \(\{F\}\) to \(\mathcal B\).
        \State \textbf{continue}.
    \EndIf

    \State Compute the standard greedy packing \(P_0\) of \(\Scal_F\) in nonincreasing original weight order.
    \State \(\lambda\gets n/(\zeta w(P_0))\).
    \State Set \(\bar w(H)\gets\lfloor\lambda w(H)\rfloor\) for every \(H\in\Scal_F\).
    \State Delete every \(H\in\Scal_F\) with \(\bar w(H)=0\).
    \State Run the two-phase core of Algorithm~\ref{alg:two phase} with weights \(\bar w\), obtaining \(P_F\).
    \State Add \(\{F\}\cup P_F\) to \(\mathcal B\).
\EndFor

\State \Return a member of \(\mathcal B\) of maximum original weight \(w\).
\end{algorithmic}
\end{minipage}
\end{adjustbox}
\end{algorithm}

\begin{proof}[Proof for the polynomial-time implementation in Theorem~\ref{thm:main}]
We first bound the number of Phase~I iterations in one residual run.  
For \(k=3\), \(L=3\), and \(\ell=7\), 
the constant \(s_0\) from Section~\ref{sec:algorithm} is exactly the exchange radius required by Thiery's Theorem~3.7.1.  
Since \(s_0\) is fixed, all independent \(X\subseteq\Scal_F\setminus A\) with \(|X|\le3s_0\) can be enumerated, 
and the condition \(\bar w^2(X)>\bar w^2(N(X,A))\) can be tested, in \(n^{3s_0+O(1)}\) time per iteration.

The scaled weights are polynomially bounded.  
The greedy packing \(P_0\) contains a maximum-weight residual set, 
so \(w(H)\le w(P_0)\) for every \(H\in\Scal_F\).  
Hence every surviving set satisfies \(1\le\bar w(H)\le\lambda w(H)\le n/\zeta\).
Every packing contains at most \(n\) sets, 
and therefore its square potential is at most \(n(n/\zeta)^2=n^3/\zeta^2\).  
Because all \(\bar w(H)\) are integers, 
every square-improving exchange increases this potential by at least \(1\). 
Thus a residual Phase~I run performs at most \(n^3/\zeta^2=n^5/(\alpha-1)^2\) iterations.

We next verify the approximation loss caused by rounding.  
Let \(O_F\) be an optimum packing of \(\Scal_F\) under the original weights.  
Removing zero-scaled sets does not change its scaled value.  
The approximation guarantee already established for the two-phase core, now applied to \(\bar w\), therefore implies
\(\alpha\bar w(P_F)\ge\bar w(O_F)\).
For every family \(Y\subseteq\Scal_F\),
\(\bar w(Y)\le\lambda w(Y)\) and
\(\bar w(Y)\ge\lambda w(Y)-|Y|\).
Consequently,
\(
\alpha w(P_F)
\ge
w(O_F)-\frac{|O_F|}{\lambda}.
\)
Since \(|O_F|\le n\) and \(1/\lambda=\zeta w(P_0)/n\), we have
\(
\frac{|O_F|}{\lambda}
\le
\zeta w(P_0)
\le
\zeta w(O_F),
\)
where the second inequality holds because \(P_0\) is feasible for the residual instance.  Hence
\[
\alpha w(P_F)\ge(1-\zeta)w(O_F).
\]

It remains to remove this factor \(1-\zeta\).  
Let \(\widehat F\) be a maximum-weight member of an optimum packing \(O\), and write \(O'=O\setminus\{\widehat F\}\).  
If \(O'=\varnothing\), then the candidate \(\{\widehat F\}\) is already optimal. 
Otherwise \(O'\subseteq\Scal_{\widehat F}\).  
Let \(O_{\widehat F}\) be an optimum packing of this residual family.  
Then \(w(O_{\widehat F})\ge w(O')\), so the residual guarantee gives \(\alpha w(P_{\widehat F})\ge(1-\zeta)w(O')\).  
Therefore
\[
\alpha\bigl(w(\widehat F)+w(P_{\widehat F})\bigr)
\ge
w(O)+(\alpha-1)w(\widehat F)-\zeta w(O').
\]
Because \(\widehat F\) is a heaviest member of \(O\) and \(|O'|\le n\),
\(w(\widehat F)\ge w(O')/n\).
Together with \(\zeta=(\alpha-1)/n\), this gives
\((\alpha-1)w(\widehat F)-\zeta w(O')\ge0\).
Thus
\[
\alpha\bigl(w(\widehat F)+w(P_{\widehat F})\bigr)\ge w(O).
\]
Since Algorithm~\ref{alg:poly-wrapper} returns the candidate of maximum original weight, 
its output has weight at least \(w(O)/\alpha\).

Finally, Phase~II is polynomial in every residual run.  
There are \(O(n^3)\) independent entering families of cardinality at most three, 
and each produces an auxiliary set of cardinality at most nine.  
Hence the auxiliary weighted \(9\)-Set Packing instance has polynomial size. 
For this fixed set size, the Phase~II routine consisting of Berman's squared-weight local search 
and the fixed scaling parameter \(\eta=1/11\) runs in polynomial time \citep{Berman2000}\citep[Section~2.5.3]{Thiery2023Thesis}.

There are at most \(n\) residual runs.  
Combining the polynomial number of Phase~I iterations, the \(n^{3s_0+O(1)}\) search time per iteration, 
and the polynomial Phase~II computation gives total running time \(n^{O(s_0)}\operatorname{poly}(\lVert I\rVert)\), 
where \(\lVert I\rVert\) denotes the binary input length.  
Since \(s_0\) is an absolute constant, the algorithm is deterministic and polynomial-time.  
Together with the approximation argument above, this proves Theorem~\ref{thm:main}.
\end{proof}

% ===============================================
\section{Conclusion}\label{sec:conclusion}
% ===============================================

We have presented a deterministic polynomial-time \(1.6908\)-approximation 
for maximum weighted \(3\)-Set Packing.  
The algorithm combines two complementary objectives.  
Phase~I uses squared-weight local search together with Thiery's layered decomposition 
to obtain the structural bound of Theorem~\ref{thm:augmented-structural}.  
The remaining deficit is measured by the positive claw gain \(\Gcal(A,O)\) under the original linear objective.  
Phase~II does not need to know the optimum claws: 
it searches over all positive-gain exchanges of size at most three, 
among which those claws form a hidden feasible benchmark.  
The Arkin-Hassin high-girth reduction is used only in the analysis, 
and Lemma~\ref{lem:crossing-gain} ensures that 
the residual gain transfers from the lifted instance back to the original one in the required direction.

As discussed in the introduction, the same \(1.6908\) approximation ratio applies to several objective-preserving special cases, 
including weighted \(3\)-Dimensional Matching, 
Cycle-only Maximum Weight Kidney Exchange with cycles of length at most three, 
Unrestricted Maximum Weight Triangle Packing, 
Maximum Weight \(3\)-Path Packing, 
and Winner Determination with single-minded bids of size at most three.  
In each case, this improves the best previously known approximation guarantee.

The result does not improve the locality gap of squared-weight local search itself.  
For \(k=3\), that locality gap remains \(\sqrt3\) even with unbounded exchange size \citep[Theorem~3.7.7]{Thiery2023Thesis}.  
Instead, the augmented structural inequality shows that 
the loss not covered by the square-potential slack can be charged to original-objective claw gain 
and then recovered by the repair phase.  
In this sense, the two objectives serve different purposes: 
the square potential provides the global decomposition and most of the slack, 
while the linear objective supplies the additional progress needed at the final singleton bottleneck.

The numerical certificates also identify the present bottleneck.  
The smallest relevant coefficients are \(\chi_1^{(7)}>0.618533561989\), \(\nu_1>0.618537374925\), 
and the final-singleton bound \(>0.6185686488\), all only slightly above \(c=0.6184\).  
By comparison, the repair-absorption margin is \((4-c)-5.5a=0.39345\).  
Thus improving only the auxiliary \(9\)-Set Packing routine would not improve the current approximation ratio; 
progress must first strengthen the Phase~I structural bounds or change the interaction between the two phases.

Several directions appear particularly natural.

\begin{itemize}
    \item \textbf{Stronger Phase~I certificates.}
    The present ratio is controlled by the level-one root and pendant coefficients together with the final-singleton certificate.  
    A natural next step is therefore to retain additional slack from the layered decomposition, 
    strengthen the constraints used in the singleton optimization, 
    or identify further structural relations among nearly tight components.  
    More substantial improvements may require modifying the nonoblivious potential itself, 
    but that would go beyond the present layered analysis.

    \item \textbf{Richer residual structures and repairs.}
    The present structural inequality exposes only the claw gain \(\Gcal(A,O)\), 
    and repairs with at most three entering sets already suffice to recover this quantity.  
    Larger repair families therefore cannot improve the current ratio by themselves.  
    They may become useful if a strengthened Phase~I analysis identifies additional residual gain 
    arising from interactions among several nearly tight claws.  
    In that case, the auxiliary repair instances would retain substantial set-system structure 
    that may permit a more effective selection procedure than treating them as arbitrary weighted Set Packing instances.

    \item \textbf{Extension to weighted \(k\)-Set Packing.}
    For fixed \(k>3\), a talon has at most \(k-1\) secondary blocker incidences after its assigned center is fixed, 
    so a claw has at most \(k(k-1)\) such incidences.  
    This suggests higher-rank analogues of the capped moment-to-mass argument used for final singletons.  
    Extending the full method is nevertheless nontrivial: the layered coefficients change with \(k\), 
    while a direct claw repair may occupy up to \(k^2\) ground elements and hence leads to a weaker auxiliary packing guarantee.  
    Determining whether objective switching can still produce a net improvement is a natural test of the framework beyond \(k=3\).

    \item \textbf{Smaller exchange neighborhoods.}
    The exchange radius required by the layered analysis is extremely large, 
    and the resulting polynomial-time algorithm is not intended to be practical.  
    It would be valuable to identify a substantially smaller canonical family of exchanges or a direct procedure 
    for detecting only the witnesses required by the structural proof.  
    Such a reduction would require showing that the restricted neighborhood still certifies every local-optimality condition 
    used in the layered decomposition, 
    but could greatly improve the running time without necessarily changing the underlying approximation mechanism.
\end{itemize}

%===============================================================================
\section*{Acknowledgements}
%===============================================================================

Both authors were supported by the Office of Research at Georgia Southern University.

\section*{Statements and Declarations}

% \textbf{Funding.}
% Both authors were supported by the Office of Research at Georgia Southern University.

% \vspace{3pt}
\noindent\textbf{Competing interests.}
The authors have no relevant financial or non-financial interests to disclose.

\vspace{3pt}
\noindent\textbf{Author contributions.}
Weitian Tong contributed to conceptualization, formal analysis, methodology, project administration, supervision, validation, writing--original draft, and writing--review and editing. 
Yao Xu contributed to formal analysis, methodology, validation, and writing--review and editing.

\vspace{3pt}
\noindent\textbf{Data availability.}
No datasets were generated or analyzed in this study.  
All information needed to support the theoretical results is contained in the manuscript and the cited references.

\vspace{3pt}
\noindent\textbf{Declaration of AI-assisted technologies.}
We used OpenAI ChatGPT to assist with language editing and 
with the refinement of verification code used 
for the certified numerical checks in Appendices~\ref{app:numeric} and~\ref{app:finite-checks}.  
The tool materially affected the presentation of the manuscript and the verification code.  
The authors independently verified the correctness and originality of all mathematical claims, proofs, computations, and references, and take full responsibility for the final contents of the paper.

\clearpage

\appendix

% ===============================================
\section{Exchange Trees and Layered Analysis}
\label{app:thiery-ideas}
% ===============================================

This appendix records the structural results from Chapters~2 and~3 of
Thiery's dissertation that are used in the main analysis \citep{Thiery2023Thesis}.  
We treat the general framework as an imported result rather than reproving it.  
Our purpose here is to state explicitly the hypotheses and specialized consequences used in the present proof:
the tree-slack inequality, the layered credit terms, 
the root and pendant coefficients, 
and the bounded-tree decomposition.  
This makes each invocation in the main proof auditable 
without introducing a second proof of the general framework.
The discussion is stated for weighted \(k\)-Set Packing when this makes the mechanism clearer; 
the main paper specializes these ingredients to \(k=3\).

% -----------------------------------------------
\subsection{Slack and interactions}
% -----------------------------------------------

Let \(A\) be a square locally optimal packing and \(O\) an optimum packing.  
As in the main paper, assign each \(o\in O\) to a maximum-weight member \(\pi(o)\in N(o,A)\), 
and define \(C_v=\{o\in O:\pi(o)=v\}\) for \(v\in A\).  
Since the conflict graph of a \(k\)-Set Packing instance is \((k+1)\)-claw free, \(|C_v|\le k\).  
Let \(N_v^+=\{v\}\cup\bigcup_{o\in C_v}N(o,A\setminus\{v\})\).
The one-claw local-optimality inequality gives \(w^2(C_v)\le w^2(N_v^+)\).

Thiery makes the slack in this inequality explicit by defining
\[
\begin{aligned}
\Delta_v & =w^2(N_v^+)-w^2(C_v) \\
\Psi_v
& =
\sum_{o\in C_v}
\left(
(w_o-w_v)^2
+
w_vw(N(o,A\setminus\{v\}))
-
w^2(N(o,A\setminus\{v\}))
\right).
\end{aligned}
\]
Both quantities are nonnegative.  The first follows from square local optimality.  
For the second, the choice of \(v=\pi(o)\) implies \(w_u\le w_v\) for every \(u\in N(o,A)\), and hence
\(w^2(N(o,A\setminus\{v\}))\le w_vw(N(o,A\setminus\{v\}))\).

Thiery's reformulation \citep[Lemma~2.3.1]{Thiery2023Thesis} of Berman's analysis then gives
\[
2w(O)
\le
w(A)
+
\sum_{o\in O}w(N(o,A))
-
\sum_{v\in A}
\left(
\frac{\Delta_v}{w_v}
+
\frac{\Psi_v}{w_v}
\right).
\]
Dropping the nonnegative slack and using
\(\sum_{o\in O}w(N(o,A))\le k\,w(A)\)
recovers Berman's \((k+1)/2\)-approximation.  
The purpose of the larger-exchange analysis is to show that substantially more of this slack can be retained.

The difficulty is that neighboring claw exchanges interact. 
Thiery records the useful interactions in the directed exchange graph \(H_\eps\) \citep[Definition~2.4.1]{Thiery2023Thesis}.  
An arc \((u,v)\) is present when some talon \(o\in C_v\) also meets \(u\) and \(w_u\ge(1-\eps)w_v\).  
Because \(v\) is a maximum-weight neighbor of \(o\), every such arc satisfies \((1-\eps)w_v\le w_u\le w_v\).  
Thus an arc records both an overlap between two claw neighborhoods and comparability of the corresponding center weights.  
Moreover, the degree of the exchange graph is at most \(k(k-1)\), 
since each of the at most \(k\) talons assigned to a center has at most \(k-1\) additional neighbors in the current packing.

% -----------------------------------------------
\subsection{High girth and trees}
\label{app:thiery-trees}
% -----------------------------------------------

The exchange graph of the original instance need not be a forest.  
Thiery therefore uses the Arkin-Hassin duplication-and-crossing reduction 
to remove bounded cycles from the conflict graph before applying the tree analysis 
\citep{ArkinHassin1998}\citep[Lemma~2.5.7]{Thiery2023Thesis}.  
The duplication-and-crossing operation and the covering interpretation of the resulting lift 
are illustrated in Figure~\ref{fig:arkin-hassin-reduction}, 
provided in Appendix~\ref{app:structural-illustrations}.  
For any fixed girth target, the construction produces a finite lifted instance in which vertex copies retain their weights, 
the bounded squared-weight local optimality required by the analysis is preserved, and the ratio \(w(O)/w(A)\) is unchanged.  
Repeating the construction permits the conflict graph to be assumed to have arbitrarily large fixed girth.  
As in the main paper, this is purely an analysis reduction; the algorithm does not construct the lift.

The relevance of high girth comes from the relation between exchange arcs and the bipartite conflict graph on \(A\cup O\).  
An arc \(u\to v\) is witnessed by a talon \(o\in C_v\) and therefore by the length-two path \(u-o-v\).  
Consequently, sufficiently short cycles or parallel arcs in a bounded exchange structure would 
create bounded cycles in the bipartite conflict graph.  
Thiery's Lemma~2.5.8 formalizes this implication: after the girth reduction, 
the bounded connected exchange structures used by the local analysis are induced trees \citep{Thiery2023Thesis}.  
This is the precise reason tree inequalities may be applied without assuming that the original input has high girth.

For a locally optimal exchange tree \(T\), write \(C_T=\bigcup_{v\in T}C_v\) and \(N_T^+=\bigcup_{v\in T}N_v^+\).  
The joint exchange inserts \(C_T\) and removes \(N_T^+\).  
Lemma~2.4.5 of \citep{Thiery2023Thesis} converts its square-local-optimality slack into the edge inequality
\[
\sum_{v\in T}\frac{\Delta_v}{w_v}
\ge
\sum_{(u,v)\in E(T)}\frac{w_u^2}{w_v}.
\]
Lemma~2.4.6 of \citep{Thiery2023Thesis} then combines these edge terms with the corresponding contributions from \(\Psi_v\), 
converting \(w_u^2/w_v\) into controlled linear weight.  
Finally, Lemma~2.4.7 of \citep{Thiery2023Thesis} distributes the edge weight over the tree.  
In particular, its weaker estimate, valid for every \(\eps\in(0,1)\), 
gives a factor \(\frac{t-1}{t}(1-\eps)\) for a tree of size \(t\).  
The factor \((t-1)/t\) reflects the \(t-1\) edges of a tree and is the basic reason larger exchange trees provide stronger guarantees.

The same construction also keeps all exchanges bounded.  
In the one-scale decomposition, large components are split until every resulting piece has longest path at most \(2\ell-2\).  
Since the exchange graph has degree at most \(k(k-1)<k^2\), Thiery obtains the safe bound
\(
|T|
\le
\sum_{r=0}^{2\ell-2} k^{2r}
\)
for each such tree \citep[Lemmas~2.5.10--2.5.11]{Thiery2023Thesis}.  
In the layered decomposition, Corollary~3.2.8 of \citep{Thiery2023Thesis} gives the corresponding bound
\[
|T|
\le
(k^2\ell)^{L-1}
\sum_{r=0}^{2\ell-2}(k\ell)^{2r},
\]
which still depends only on the fixed parameters \(k\), \(\ell\), and \(L\), and not on the input size.

% -----------------------------------------------
\subsection{Layered hierarchy}
% -----------------------------------------------

Chapter~2 of \citep{Thiery2023Thesis} first applies this tree mechanism at a single weight scale.  
Large connected components contain sufficiently many comparable-weight interactions to be partitioned into 
locally optimal exchange trees of size at least \(\ell\).  The resulting tree contribution is controlled by
\(\chi_\eps^{(\ell)}=\frac{\ell-1}{\ell}(1-\eps)\).
As \(\ell\) grows, this coefficient approaches \(1-\eps\).

The complementary case consists of small components that remain isolated from other comparable-weight interactions.  
Their isolation forces stronger separation between the weights of their centers and external neighbors, 
producing additional \(\Psi\)-slack.  
Writing \(1-\eps=(1-\delta)^2\), Thiery captures this contribution through
\(\rho_{k,\delta}=k\delta-\delta/(1-\delta)\).
The one-scale guarantee is determined by the smaller of the tree and isolation coefficients.  
For \(k=3\), optimizing this tradeoff gives the \(1.761+\eps\) approximation for arbitrarily small fixed \(\eps>0\) 
developed in Chapter~2 \citep[Theorems~2.5.5--2.5.6]{Thiery2023Thesis}.

Chapter~3 of \citep{Thiery2023Thesis} refines this dichotomy by observing that isolation depends on the weight scale. 
The potential exchange arcs are partitioned into layers \(H_1,\ldots,H_{L+1}\) according to their endpoint weight ratios.  
Let \(H_{\le i}=H_1\cup\cdots\cup H_i\), with thresholds satisfying \(1-\eps_i=(1-\eps_{i-1})^2\). 
A component that remains unresolved at one layer is reconsidered after the weight threshold is relaxed.  
It either acquires new controlled interactions or remains separated for another layer,
in which case stronger weight-ratio information becomes available.

The recursive construction is implemented through contraction.  
At each layer, previously constructed trees and unresolved components are contracted, 
the next class of exchange arcs is added, 
and the new connected structures are classified according to whether they become large, 
attach to an earlier tree, or remain unresolved.  
Corollary~3.2.8 of \citep{Thiery2023Thesis} decomposes \(A\) into vertex-disjoint trees, 
and Definition~3.2.9 of \citep{Thiery2023Thesis} distinguishes within each such tree 
a \emph{root tree} and its attached \emph{pendant trees} \citep{Thiery2023Thesis}.  
Every pendant tree is joined to the previously constructed part by a distinguished \emph{bridge arc}.  
The layer in which a piece becomes attached records how long it remained isolated 
and therefore controls the weight ratios of its remaining external neighbors through Thiery's separation lemmas.

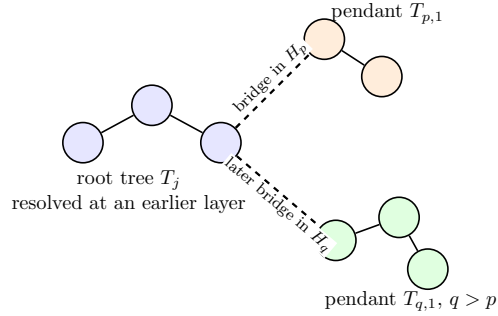
\begin{figure}[htbp]
\centering
\begin{adjustbox}{max width=.4\textwidth}
\begin{tikzpicture}[
    rootv/.style={circle,draw,thick,fill=blue!10,minimum size=7mm},
    p1/.style={circle,draw,thick,fill=orange!14,minimum size=7mm},
    p2/.style={circle,draw,thick,fill=green!12,minimum size=7mm},
    bridge/.style={very thick,dashed},
    normal/.style={thick},
    lab/.style={font=\small,align=center}
]
    \node[rootv] (r1) at (-1.2,0) {};
    \node[rootv] (r2) at (0,0.65) {};
    \node[rootv] (r3) at (1.2,0) {};
    \draw[normal] (r1)--(r2)--(r3);
    \node[lab] at (-0.4,-0.85)
        {root tree \(T_j\)\\resolved at an earlier layer};

    \node[p1] (a1) at (3.0,1.8) {};
    \node[p1] (a2) at (4.0,1.15) {};
    \draw[normal] (a1)--(a2);
    \draw[bridge]
        (r3)--node[above,sloped,pos=.58,font=\scriptsize,fill=white,inner sep=1pt]
        {bridge in \(H_p\)} (a1);
    \node[lab] at (4.1,2.25) {pendant \(T_{p,1}\)};

    \node[p2] (b1) at (3.2,-1.7) {};
    \node[p2] (b2) at (4.3,-1.3) {};
    \node[p2] (b3) at (4.8,-2.2) {};
    \draw[normal] (b1)--(b2)--(b3);
    \draw[bridge]
        (r3)--node[below,sloped,pos=.58,font=\scriptsize,fill=white,inner sep=1pt]
        {later bridge in \(H_q\)} (b1);
    \node[lab] at (4.5,-2.75)
        {pendant \(T_{q,1}\), \(q>p\)};

\end{tikzpicture}
\end{adjustbox}
\caption{Schematic root and pendant structure in Thiery's layered decomposition.  
Earlier trees are contracted rather than rebuilt as the layers are processed, 
and each pendant tree records the layer in which it first becomes attached through its bridge arc.  
This attachment history determines the weight-separation information available for its remaining external neighbors.  
After the Arkin-Hassin high-girth reduction, the complete induced exchange structure is a tree, 
while the root, pendant, and bridge labels retain the history needed by the layered analysis.}
\label{fig:hierarchical tree}
\end{figure}

Figure~\ref{fig:hierarchical tree} illustrates why the layered decomposition contains more information than 
a decomposition formed only from the final union \(H_{\le L}\).  
If a component remains unresolved through a given layer, 
the absence of earlier exchange arcs certifies separation from the relevant external neighbors.  
Lemmas~3.2.10--3.2.11 of \citep{Thiery2023Thesis} translate this history into explicit weight-ratio bounds.  
The resulting tree therefore records both its controlled internal interactions 
and the separation information accumulated before its pieces became attached.

% -----------------------------------------------
\subsection{Global coefficients}
% -----------------------------------------------

The layered analysis combines internal and external slack.  
Internal root and pendant edges are charged through the joint exchange inequalities, 
whereas external neighbors are controlled by the weight separation inherited from the layer structure.  
The remaining exterior terms are collected into layer credits.  
These credits are bookkeeping terms rather than an additional source of approximation: 
after summing the component inequalities, 
they cancel the corresponding portions of the positive neighborhood term in Berman's master inequality.

For an ordinary root tree of size \(t\) at level \(j\), Thiery obtains
\[
\chi_{j-1}^{(t)}
=
\rho_{k,\eps_{j-1}}
+
\frac{t-1}{t}(1-\eps_{j-1})^5,
\]
whereas the first layer has
\(\chi_{-1}^{(t)}=\frac{t-1}{t}(1-\eps_0)^2\).
For a pendant tree at level \(j\), the coefficient is
\[
\nu_{j-1}
=
\rho_{k,\eps_{j-1}}
+
(1-\eps_{j-1})^5
\left[
1-(1-\eps_{j-1})^2\eps_{j-1}
\right].
\]
A component that remains unresolved through the final layer is controlled in the uniform analysis by \(\rho_{k,\eps_{L-1}}\); 
see Lemmas~3.5.1, 3.5.2, and~3.6.1 of \citep{Thiery2023Thesis}.

After the layer credits are summed and canceled, Theorem~3.7.1 of \citep{Thiery2023Thesis} gives
\[
w(O)
\le
\frac12
\left[
k+1
-
\min
\left\{
\{\chi_{i-1}^{(\ell)}\}_{i=0}^{L-1},
\{\nu_{i-1}\}_{i=1}^{L-1},
\rho_{k,\eps_{L-1}}
\right\}
\right]
w(A).
\]
The theorem assumes local optimality under \(s\)-exchanges for
\[
s
\ge
(k^2\ell)^L
\sum_{r=0}^{2\ell-2}(k\ell)^{2r}.
\]
For the parameters \(k=3\), \(L=3\), and \(\ell=7\) used in the main paper, 
this is exactly the constant \(s_0\) fixed in Section~\ref{sec:algorithm}.

For \(k=3\), Thiery chooses the layer parameters so that the final uniform coefficient is
\(\rho_{3,1-1/\sqrt3}=4-2\sqrt3\)
and all root and pendant coefficients are larger.  
Substitution into Theorem~3.7.1 gives the \(\sqrt3\)-approximation \citep[Theorem~3.7.7]{Thiery2023Thesis}.  
The tight final structure is an isolated singleton center with three talons, 
corresponding to the square-locality-gap configuration shown in Figure~\ref{fig:sqrt3-barrier}.

The present paper keeps this layered decomposition but strengthens the analysis precisely 
where the uniform \(\sqrt3\) bound loses information.  
For nonsingleton final components, we retain internal-edge slack that is discarded by the uniform final coefficient.  
For a final singleton, we exploit the set-system structure specific to \(3\)-Set Packing: 
each talon contains at most three ground elements and, 
because the sets in the current packing \(A\) are pairwise disjoint, 
it can intersect at most three members of \(A\).  
Once its assigned center is fixed, each talon therefore has at most two additional blocker incidences.  
Combining this restriction with square local optimality yields the stronger finite certificate of Lemma~\ref{lem:singleton-certificate}.  
The residual positive claw gain is then recovered under the original weight objective by Phase~II.

% ===============================================
\section{Certified Arithmetic}
\label{app:numeric}
% ===============================================

This appendix records the numerical comparisons used in the structural analysis.  
All constants entering the proof are rational, 
and the displayed radical bounds can be certified by exact rational squaring.  
The decimal values below are therefore used only for readability 
and do not introduce any dependence on floating-point arithmetic.

Recall that
\(a=\frac{5433}{10000}\),
\(c=\frac{773}{1250}\), and 
\(q=a^2=\frac{29517489}{10^8}=0.29517489\).
The radicals appearing in the nonfinal coefficients satisfy
\[
0.7370888684548<\sqrt a<0.7370888684549,
\qquad
0.8585387984563<a^{1/4}<0.8585387984564.
\]
The first interval is verified by squaring its endpoints as exact rationals and comparing with \(a\); 
the second is verified in the same way by taking fourth powers.  
Substituting these certified intervals into the coefficient formulas of Table~\ref{tab:coefficients} gives, in particular,
\[
\chi_1^{(7)}>0.618533561989,
\qquad
\nu_1>0.618537374925.
\]
Hence the two smallest imported nonfinal coefficients both exceed \(c=0.6184\), 
with the smaller margin greater than \(1.33\times10^{-4}\).

For the active final-singleton calculation, let
\[
D=9q^2+18q-3
=\frac{30973019611786089}{10^{16}}
=3.0973019611786089.
\]
Exact squaring gives
\(
1.7599153278435<\sqrt D<1.7599153278436.
\)
Therefore the smaller root
\(
x_*=
\frac{3(1+3q)-\sqrt D}{2}
\)
satisfies
\(
1.9483293410782<x_*<1.9483293410783.
\)
Since \(4-a-(2-a)x\) is decreasing in \(x\), substitution of these certified endpoints gives
\[
0.6185686488
<
4-a-(2-a)x_*
<
0.6185686490.
\]
Thus the active final-singleton value exceeds \(c\) by more than \(1.68\times10^{-4}\).

The absorption comparison used when the two phases are combined is exact.  
Since the auxiliary approximation factor is \(11/2\),
\(
\frac{11}{2}a
=
\frac{59763}{20000}
=
2.98815\) and 
\(
4-c
=
\frac{67632}{20000}
=
3.3816.
\)
Consequently,
\[
(4-c)-\frac{11}{2}a
=
\frac{7869}{20000}
=
0.39345.
\]
The repair step therefore has substantially more numerical slack than the Phase~I structural inequalities.  
The next section gives the complete finite verification of the final-singleton program.

% ===============================================
\section{Singleton Finite Check}
\label{app:finite-checks}
% ===============================================

We verify the finite optimization used in Lemma~\ref{lem:singleton-certificate}.  
Recall that \(m=|C_v|\in\{1,2,3\}\), \(x=\sum_{o\in C_v}w_o/w_v\), and
\(
S_m(x)=\left[\frac{x^2}{m}-1\right]_+.
\)
With \(q=a^2\), the minimum possible normalized blocker mass compatible with squared blocker mass \(S\) is
\(
\beta_{\min}(S)=hq+\sqrt{S-hq^2}
\)
whenever \(hq^2\le S\le(h+1)q^2\).  By Lemma~\ref{lem:singleton-relations} (c),
there are at most \(2m\) secondary blocker coordinates, each bounded above by \(q\); 
hence feasibility requires \(0\le S_m(x)\le2mq^2\).

The reduced objective is
\[
\Phi_m(x)
=
m+1-2x+(3-m)(1-a)
+a\max\{\beta_{\min}(S_m(x)),x-1\}.
\]
No independent upper bound such as \(x\le m\) is used: 
the assignment of a talon to a maximum-weight member of its neighborhood compares weights only within the current packing \(A\), 
and does not imply \(w_o\le w_v\).  
The feasible range of \(x\) is instead determined by the square-locality constraint \(S_m(x)\le2mq^2\).

When \(S_m(x)=0\), we have \(\beta_{\min}(S_m(x))=0\).  
On this region, \(\Phi_m\) is piecewise linear because of \(\max\{0,x-1\}\).  
Its slopes are \(-2\) and \(-2+a\), both negative, so its minimum is attained at the right endpoint \(x=\sqrt m\).

For \(S_m(x)>0\), define the square-mass breakpoints
\(
x_{m,h}=\sqrt{m(1+hq^2)}\),
\(
h=0,1,\ldots,2m.
\)
On the branch \(hq^2\le S_m(x)\le(h+1)q^2\),
\[
\beta_{\min}(S_m(x))
=
hq+\sqrt{\frac{x^2}{m}-1-hq^2}.
\]
The square-root term is concave in \(x\) on its domain.  
After splitting a branch at every crossing
\(\beta_{\min}(S_m(x))=x-1\),
each resulting piece of \(\Phi_m\) is therefore either linear or concave, 
and its minimum is attained at an endpoint.

For completeness, a crossing on branch \(h\) must satisfy
\[
hq+\sqrt{\frac{x^2}{m}-1-hq^2}=x-1.
\]
After isolating the square root, the necessary sign condition is \(x-1-hq\ge0\).  Squaring then gives
\[
\left(1-\frac1m\right)x^2
-2(1+hq)x
+(1+hq)^2+1+hq^2
=0.
\]
A root is admissible only if it also satisfies
\(hq^2\le S_m(x)\le(h+1)q^2\).
Checking these conditions leaves exactly two interior crossings: one for \(m=2,h=1\) and one for \(m=3,h=3\).  
For \(m=1\), the only admissible crossing is the boundary point \(x=1\).

The resulting complete candidate list is shown in Table~\ref{tab:singleton-finite-check}.  
The displayed values are rounded for readability; the certified comparisons use higher precision.

\begin{table}[htbp]
\centering
\caption{Candidate points in the finite verification of the final-singleton program.  
Bold entries are the minimum values for the corresponding number \(m\) of talons.}
\label{tab:singleton-finite-check}
\small
\begin{adjustbox}{max width=.5\textwidth}
\begin{tabular}{@{}clrr@{}}
\toprule
\(m\) & candidate & \(x\) & \(\Phi_m(x)\) \\
\midrule
1 & \(S_1=0\) endpoint & \(1.0000000000\) & \(\mathbf{0.9134000000}\) \\
1 & \(h=1\) breakpoint & \(1.0426544086\) & \(0.9884597006\) \\
1 & \(h=2\) endpoint & \(1.0836311325\) & \(1.0668747704\) \\
\midrule
2 & \(S_2=0\) endpoint & \(1.4142135624\) & \(0.8533151037\) \\
2 & \(h=1\) breakpoint & \(1.4745360055\) & \(0.7654434009\) \\
2 & crossing on \(h=1\) & \(1.5037497800\) & \(0.7228876955\) \\
2 & \(h=2\) breakpoint & \(1.5324858442\) & \(\mathbf{0.7124653470}\) \\
2 & \(h=3\) breakpoint & \(1.5883227928\) & \(0.7611599676\) \\
2 & \(h=4\) endpoint & \(1.6422623802\) & \(0.8136493105\) \\
\midrule
3 & \(S_3=0\) endpoint & \(1.7320508076\) & \(0.9336215886\) \\
3 & \(h=1\) breakpoint & \(1.8059304104\) & \(0.8260011712\) \\
3 & \(h=2\) breakpoint & \(1.8769041782\) & \(0.7226136836\) \\
3 & \(h=3\) breakpoint & \(1.9452901946\) & \(0.6229957735\) \\
3 & crossing on \(h=3\) & \(1.9483293411\) & \(\mathbf{0.6185686489}\) \\
3 & \(h=4\) breakpoint & \(2.0113524277\) & \(0.6187692156\) \\
3 & \(h=5\) breakpoint & \(2.0753128042\) & \(0.6512169803\) \\
3 & \(h=6\) endpoint & \(2.1373600264\) & \(0.6874910537\) \\
\bottomrule
\end{tabular}
\end{adjustbox}
\end{table}

For \(m=1\), the minimum is attained at \(x=1\) and equals \(2(1-a)=0.9134\).  
The two larger candidate points are feasible square-mass breakpoints and are included in the table; 
neither is competitive with the value at \(x=1\).

For \(m=2\), the minimum occurs at \(S_2(x)=2q^2\), namely
\(x=\sqrt{2(1+2q^2)}=\sqrt{2(1+2a^4)}\), and
\[
\Phi_2(x)
=
4-a-2\sqrt{2(1+2a^4)}+2a^3
>
0.712465347026.
\]

For \(m=3\), the minimum occurs at the interior crossing on the \(h=3\) branch.  There
\[
3q+\sqrt{\frac{x^2}{3}-1-3q^2}=x-1,
\]
and squaring yields
\(
x^2-3(1+3q)x+3+9q+18q^2=0.
\)
The smaller root is
\[
x_*=
\frac{3(1+3q)-\sqrt{9q^2+18q-3}}{2}.
\]
The certified interval above gives \(1.9483293410782<x_*<1.9483293410783\).  Moreover,
\(
3<
\frac{x_*^2/3-1}{q^2}
<
4
\)
and \(x_*-1-3q>0\), so the root lies strictly inside the \(h=3\) square-mass branch and satisfies the unsquared sign condition.  
At this crossing,
\[
\Phi_3(x_*)
=
4-a-(2-a)x_*
>
0.6185686488
>
c.
\]

Thus the three minima satisfy
\[
\min\Phi_1=0.9134,\qquad
\min\Phi_2>0.712465347026,\qquad
\min\Phi_3>0.6185686488>c=0.6184.
\]
This completes the certified finite verification required by Lemma~\ref{lem:singleton-certificate}.

% ===============================================
\section{Structural Illustrations}
\label{app:structural-illustrations}
% ===============================================

This appendix provides schematic illustrations of two structural constructions used in the analysis: 
the Arkin-Hassin high-girth reduction and Thiery's recursive exchange layers.  
Figure~\ref{fig:arkin-hassin-reduction} illustrates the Arkin-Hassin duplication-and-crossing reduction discussed in Section~\ref{sec:prelim}.  
Figure~\ref{fig:thiery-exchange-layers} illustrates the recursive construction of Thiery's exchange-layer classes described in Section~\ref{subsec:thiery-template}.

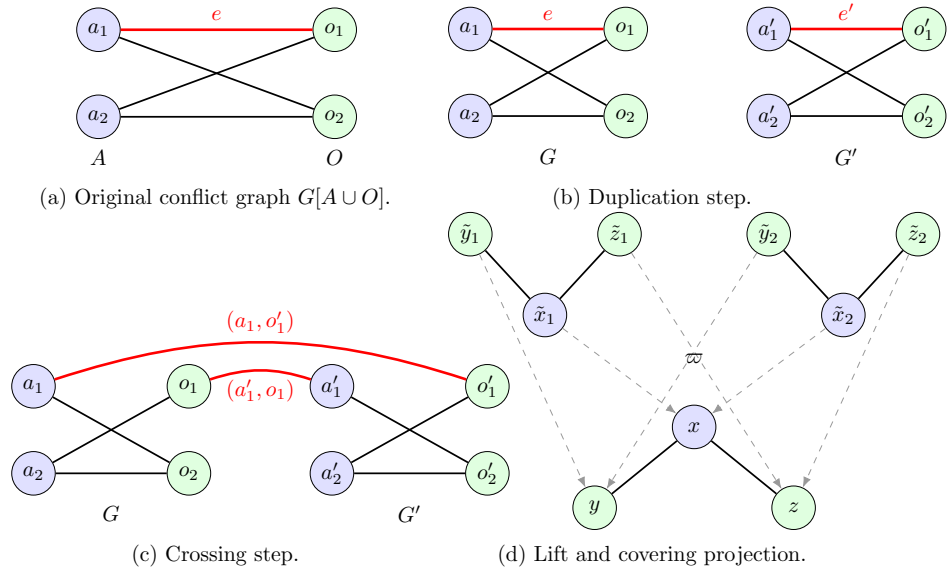
\begin{figure}[htbp]
\centering
\begin{adjustbox}{max width=.7\textwidth}
\begin{tabular}{cc}

%====================================================
% Subfigure (a)
%====================================================
\begin{subfigure}[t]{0.4\textwidth}
\centering
\begin{tikzpicture}[
    font=\small,
    >=Latex,
    Avertex/.style={circle, draw, fill=blue!12, minimum size=7mm, inner sep=0pt},
    Overtex/.style={circle, draw, fill=green!12, minimum size=7mm, inner sep=0pt},
    edge/.style={thick},
    chosen/.style={very thick, red}
]
    \node[Avertex] (a1) at (0,1.4) {\(a_1\)};
    \node[Avertex] (a2) at (0,0)   {\(a_2\)};
    \node[Overtex] (o1) at (3.8,1.4) {\(o_1\)};
    \node[Overtex] (o2) at (3.8,0)   {\(o_2\)};

    \draw[chosen] (a1) -- node[above] {\(e\)} (o1);
    \draw[edge] (a1) -- (o2);
    \draw[edge] (a2) -- (o1);
    \draw[edge] (a2) -- (o2);

    \node at (0,-0.65) {\(A\)};
    \node at (3.8,-0.65) {\(O\)};
\end{tikzpicture}
\caption{Original conflict graph \(G[A\cup O]\).}\label{fig:ah-a}
\end{subfigure}
&
%====================================================
% Subfigure (b)
%====================================================
\begin{subfigure}[t]{0.4\textwidth}
\centering
\begin{tikzpicture}[
    font=\small,
    >=Latex,
    Avertex/.style={circle, draw, fill=blue!12, minimum size=7mm, inner sep=0pt},
    Overtex/.style={circle, draw, fill=green!12, minimum size=7mm, inner sep=0pt},
    edge/.style={thick},
    chosen/.style={very thick, red}
]
    % left copy
    \node[Avertex] (a1) at (0,1.4) {\(a_1\)};
    \node[Avertex] (a2) at (0,0)   {\(a_2\)};
    \node[Overtex] (o1) at (2.5,1.4) {\(o_1\)};
    \node[Overtex] (o2) at (2.5,0)   {\(o_2\)};

    \draw[chosen] (a1) -- node[above] {\(e\)} (o1);
    \draw[edge] (a1) -- (o2);
    \draw[edge] (a2) -- (o1);
    \draw[edge] (a2) -- (o2);

    % right copy
    \node[Avertex] (a1p) at (4.8,1.4) {\(a_1'\)};
    \node[Avertex] (a2p) at (4.8,0)   {\(a_2'\)};
    \node[Overtex] (o1p) at (7.3,1.4) {\(o_1'\)};
    \node[Overtex] (o2p) at (7.3,0)   {\(o_2'\)};

    \draw[chosen] (a1p) -- node[above] {\(e'\)} (o1p);
    \draw[edge] (a1p) -- (o2p);
    \draw[edge] (a2p) -- (o1p);
    \draw[edge] (a2p) -- (o2p);

    \node at (1.25,-0.65) {\(G\)};
    \node at (6.05,-0.65) {\(G'\)};
\end{tikzpicture}
\caption{Duplication step.}\label{fig:ah-b}
\end{subfigure}
\\[1.2em]

%====================================================
% Subfigure (c)
%====================================================
\begin{subfigure}[t]{0.4\textwidth}
\centering
\begin{tikzpicture}[
    font=\small,
    >=Latex,
    Avertex/.style={circle, draw, fill=blue!12, minimum size=7mm, inner sep=0pt},
    Overtex/.style={circle, draw, fill=green!12, minimum size=7mm, inner sep=0pt},
    edge/.style={thick},
    chosen/.style={very thick, red}
]
    % left copy
    \node[Avertex] (a1) at (0,1.4) {\(a_1\)};
    \node[Avertex] (a2) at (0,0)   {\(a_2\)};
    \node[Overtex] (o1) at (2.5,1.4) {\(o_1\)};
    \node[Overtex] (o2) at (2.5,0)   {\(o_2\)};

    % right copy
    \node[Avertex] (a1p) at (4.8,1.4) {\(a_1'\)};
    \node[Avertex] (a2p) at (4.8,0)   {\(a_2'\)};
    \node[Overtex] (o1p) at (7.3,1.4) {\(o_1'\)};
    \node[Overtex] (o2p) at (7.3,0)   {\(o_2'\)};

    % unchanged edges
    \draw[edge] (a1) -- (o2);
    \draw[edge] (a2) -- (o1);
    \draw[edge] (a2) -- (o2);

    \draw[edge] (a1p) -- (o2p);
    \draw[edge] (a2p) -- (o1p);
    \draw[edge] (a2p) -- (o2p);

    % crossed edges: clearly two separate edges
    \draw[chosen] (a1) to[bend left=18] node[above] {\((a_1,o_1')\)} (o1p);
    \draw[chosen] (a1p) to[bend right=18] node[below] {\((a_1',o_1)\)} (o1);

    \node at (1.25,-0.65) {\(G\)};
    \node at (6.05,-0.65) {\(G'\)};
\end{tikzpicture}
\caption{Crossing step.}\label{fig:ah-c}
\end{subfigure}
&
%====================================================
% Subfigure (d)
%====================================================
\begin{subfigure}[t]{0.4\textwidth}
\centering
\begin{tikzpicture}[
    font=\small,
    >=Latex,
    Avertex/.style={circle, draw, fill=blue!12, minimum size=7mm, inner sep=0pt},
    Overtex/.style={circle, draw, fill=green!12, minimum size=7mm, inner sep=0pt},
    edge/.style={thick},
    proj/.style={->, dashed, gray!80}
]
    % lift
    \node[Avertex] (tx1) at (1.2,2.2) {\(\tilde x_1\)};
    \node[Avertex] (tx2) at (6.0,2.2) {\(\tilde x_2\)};

    \node[Overtex] (ty1) at (0,3.5) {\(\tilde y_1\)};
    \node[Overtex] (tz1) at (2.4,3.5) {\(\tilde z_1\)};
    \node[Overtex] (ty2) at (4.8,3.5) {\(\tilde y_2\)};
    \node[Overtex] (tz2) at (7.2,3.5) {\(\tilde z_2\)};

    \draw[edge] (tx1) -- (ty1);
    \draw[edge] (tx1) -- (tz1);
    \draw[edge] (tx2) -- (ty2);
    \draw[edge] (tx2) -- (tz2);

    % original
    \node[Avertex] (x) at (3.6,0.4) {\(x\)};
    \node[Overtex] (y) at (2.0,-0.9) {\(y\)};
    \node[Overtex] (z) at (5.2,-0.9) {\(z\)};

    \draw[edge] (x) -- (y);
    \draw[edge] (x) -- (z);

    % projections
    \draw[proj] (tx1) -- (x);
    \draw[proj] (tx2) -- (x);
    \draw[proj] (ty1) -- (y);
    \draw[proj] (ty2) -- (y);
    \draw[proj] (tz1) -- (z);
    \draw[proj] (tz2) -- (z);

    \node at (3.6,1.5) {\(\varpi\)};
\end{tikzpicture}
\caption{Lift and covering projection.}\label{fig:ah-d}
\end{subfigure}

\end{tabular}
\end{adjustbox}
\caption{The Arkin-Hassin duplication-and-crossing reduction. 
In Subfigure~\ref{fig:ah-a}, the conflict graph \(G[A\cup O]\) contains a shortest cycle, and an edge \(e\) on such a cycle is selected. 
In Subfigure~\ref{fig:ah-b}, the graph is duplicated, producing a copy \(G'\) and the corresponding edge \(e'\). 
In Subfigure~\ref{fig:ah-c}, the edges \(e=(a_1,o_1)\) and \(e'=(a_1',o_1')\) are removed and replaced by the crossed edges \((a_1,o_1')\) and \((a_1',o_1)\). 
Applying this operation systematically to all copies of every current shortest cycle eliminates those cycles without introducing shorter ones; iterating the process yields a lift of arbitrarily large prescribed girth while preserving the local-optimality conditions used in the analysis and the ratio \(w(O)/w(A)\). 
Subfigure~\ref{fig:ah-d} illustrates the covering viewpoint for one duplication step, where \(M=2\): each original vertex has two copies, and each copy has exactly one neighbor above each neighbor of the original vertex.  After \(r=1\) duplication steps, the same property holds with \(M=2^r\) copies of each original vertex.
}
\label{fig:arkin-hassin-reduction}
\end{figure}

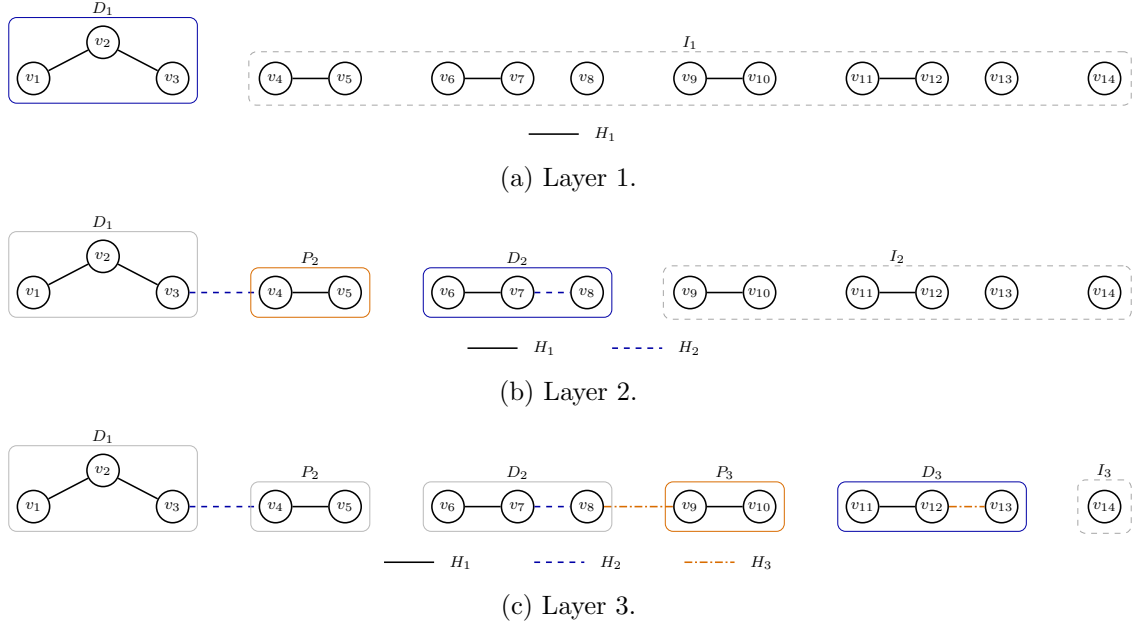
\begin{figure}[t]
\centering

%====================================================
% Common styles
%====================================================
\tikzset{
  exchvertex/.style={
    circle,
    draw,
    thick,
    fill=white,
    minimum size=5.8mm,
    inner sep=0pt,
    font=\scriptsize
  },
  hOne/.style={
    thick,
    draw=black
  },
  hTwo/.style={
    thick,
    dashed,
    draw=blue!65!black
  },
  hThree/.style={
    thick,
    draw=orange!85!black,
    dash pattern=on 4pt off 1.4pt on 1pt off 1.4pt
  },
  resolvedbox/.style={
    draw=gray!55,
    rounded corners,
    inner sep=4pt
  },
  currentbox/.style={
    draw=blue!65!black,
    rounded corners,
    inner sep=4pt
  },
  pendantbox/.style={
    draw=orange!85!black,
    rounded corners,
    inner sep=4pt
  },
  unresolvedbox/.style={
    draw=gray!70,
    dashed,
    rounded corners,
    inner sep=5pt
  },
  classlabel/.style={
    font=\scriptsize,
    fill=white,
    inner sep=1pt
  }
}

%====================================================
% Subfigure (a): Layer 1
%====================================================
\begin{subfigure}[t]{\textwidth}
\centering
\begin{adjustbox}{max width=.9\textwidth}
\begin{tikzpicture}[x=1cm,y=1cm]

% Fixed vertex set A
\foreach \i/\x/\y in {
  1/0/0,
  2/1.25/0.65,
  3/2.50/0,
  4/4.35/0,
  5/5.60/0,
  6/7.45/0,
  7/8.70/0,
  8/9.95/0,
  9/11.80/0,
  10/13.05/0,
  11/14.90/0,
  12/16.15/0,
  13/17.40/0,
  14/19.25/0
}{
  \node[exchvertex] (v\i) at (\x,\y) {$v_{\i}$};
}

% H_1 edges
\draw[hOne] (v1)--(v2)--(v3);
\draw[hOne] (v4)--(v5);
\draw[hOne] (v6)--(v7);
\draw[hOne] (v9)--(v10);
\draw[hOne] (v11)--(v12);

% Classes
\node[currentbox,fit=(v1)(v2)(v3)] (D1) {};
\node[classlabel,anchor=south] at (D1.north) {$D_1$};

\node[
  unresolvedbox,
  fit=(v4)(v5)(v6)(v7)(v8)(v9)(v10)(v11)(v12)(v13)(v14)
] (I1) {};
\node[classlabel,anchor=south] at (I1.north) {$I_1$};

% Edge legend
\draw[hOne] (8.9,-1.00)--(9.8,-1.00);
\node[font=\scriptsize,anchor=west] at (9.95,-1.00) {$H_1$};

\end{tikzpicture}
\end{adjustbox}
\caption{Layer 1.}
\label{fig:thiery-layers-a}
\end{subfigure}

\par\medskip

%====================================================
% Subfigure (b): Layer 2
%====================================================
\begin{subfigure}[t]{\textwidth}
\centering
\begin{adjustbox}{max width=.9\textwidth}
\begin{tikzpicture}[x=1cm,y=1cm]

% Same fixed vertex set A
\foreach \i/\x/\y in {
  1/0/0,
  2/1.25/0.65,
  3/2.50/0,
  4/4.35/0,
  5/5.60/0,
  6/7.45/0,
  7/8.70/0,
  8/9.95/0,
  9/11.80/0,
  10/13.05/0,
  11/14.90/0,
  12/16.15/0,
  13/17.40/0,
  14/19.25/0
}{
  \node[exchvertex] (v\i) at (\x,\y) {$v_{\i}$};
}

% H_1 edges
\draw[hOne] (v1)--(v2)--(v3);
\draw[hOne] (v4)--(v5);
\draw[hOne] (v6)--(v7);
\draw[hOne] (v9)--(v10);
\draw[hOne] (v11)--(v12);

% New H_2 edges
\draw[hTwo] (v3)--(v4);
\draw[hTwo] (v7)--(v8);

% Previously resolved class
\node[resolvedbox,fit=(v1)(v2)(v3)] (D1) {};
\node[classlabel,anchor=south] at (D1.north) {$D_1$};

% P_2
\node[pendantbox,fit=(v4)(v5)] (P2) {};
\node[classlabel,anchor=south] at (P2.north) {$P_2$};

% D_2
\node[currentbox,fit=(v6)(v7)(v8)] (D2) {};
\node[classlabel,anchor=south] at (D2.north) {$D_2$};

% I_2
\node[
  unresolvedbox,
  fit=(v9)(v10)(v11)(v12)(v13)(v14)
] (I2) {};
\node[classlabel,anchor=south] at (I2.north) {$I_2$};

% Edge legend
\draw[hOne] (7.8,-1.00)--(8.7,-1.00);
\node[font=\scriptsize,anchor=west] at (8.85,-1.00) {$H_1$};

\draw[hTwo] (10.4,-1.00)--(11.3,-1.00);
\node[font=\scriptsize,anchor=west] at (11.45,-1.00) {$H_2$};

\end{tikzpicture}
\end{adjustbox}
\caption{Layer 2.}
\label{fig:thiery-layers-b}
\end{subfigure}

\par\medskip

%====================================================
% Subfigure (c): Layer 3
%====================================================
\begin{subfigure}[t]{\textwidth}
\centering
\begin{adjustbox}{max width=.9\textwidth}
\begin{tikzpicture}[x=1cm,y=1cm]

% Same fixed vertex set A
\foreach \i/\x/\y in {
  1/0/0,
  2/1.25/0.65,
  3/2.50/0,
  4/4.35/0,
  5/5.60/0,
  6/7.45/0,
  7/8.70/0,
  8/9.95/0,
  9/11.80/0,
  10/13.05/0,
  11/14.90/0,
  12/16.15/0,
  13/17.40/0,
  14/19.25/0
}{
  \node[exchvertex] (v\i) at (\x,\y) {$v_{\i}$};
}

% H_1 edges
\draw[hOne] (v1)--(v2)--(v3);
\draw[hOne] (v4)--(v5);
\draw[hOne] (v6)--(v7);
\draw[hOne] (v9)--(v10);
\draw[hOne] (v11)--(v12);

% H_2 edges
\draw[hTwo] (v3)--(v4);
\draw[hTwo] (v7)--(v8);

% New H_3 edges
\draw[hThree] (v8)--(v9);
\draw[hThree] (v12)--(v13);

% Earlier classes
\node[resolvedbox,fit=(v1)(v2)(v3)] (D1) {};
\node[classlabel,anchor=south] at (D1.north) {$D_1$};

\node[resolvedbox,fit=(v4)(v5)] (P2) {};
\node[classlabel,anchor=south] at (P2.north) {$P_2$};

\node[resolvedbox,fit=(v6)(v7)(v8)] (D2) {};
\node[classlabel,anchor=south] at (D2.north) {$D_2$};

% P_3
\node[pendantbox,fit=(v9)(v10)] (P3) {};
\node[classlabel,anchor=south] at (P3.north) {$P_3$};

% D_3
\node[currentbox,fit=(v11)(v12)(v13)] (D3) {};
\node[classlabel,anchor=south] at (D3.north) {$D_3$};

% I_3
\node[unresolvedbox,fit=(v14)] (I3) {};
\node[classlabel,anchor=south] at (I3.north) {$I_3$};

% Edge legend
\draw[hOne] (6.3,-1.00)--(7.2,-1.00);
\node[font=\scriptsize,anchor=west] at (7.35,-1.00) {$H_1$};

\draw[hTwo] (9.0,-1.00)--(9.9,-1.00);
\node[font=\scriptsize,anchor=west] at (10.05,-1.00) {$H_2$};

\draw[hThree] (11.7,-1.00)--(12.6,-1.00);
\node[font=\scriptsize,anchor=west] at (12.75,-1.00) {$H_3$};

\end{tikzpicture}
\end{adjustbox}
\caption{Layer 3.}
\label{fig:thiery-layers-c}
\end{subfigure}

\caption{Illustration of Thiery's recursive exchange-layer partition on a fixed vertex set \(A\), with \(\ell=3\). Solid, dashed, and dash-dotted edges denote arcs of \(H_1\), \(H_2\), and \(H_3\), respectively; orientations are suppressed because the partition uses only connectivity. In Subfigure~\ref{fig:thiery-layers-a}, \(D_1=\{v_1,v_2,v_3\}\), and the remaining vertices form \(I_1\). In Subfigure~\ref{fig:thiery-layers-b}, the partition is refined on \(I_1\): \(\{v_6,v_7,v_8\}\) is a size-\(\ell\) component of \(H_{\le2}[I_1]\) and hence forms \(D_2\), whereas \(\{v_4,v_5\}\) remains a component of size less than \(\ell\) in \(H_{\le2}[I_1]\) but belongs to a larger component of the full \(H_{\le2}\) through the \(H_2\)-edge \(v_3v_4\), and hence forms \(P_2\). The remaining vertices form \(I_2\). In Subfigure~\ref{fig:thiery-layers-c}, the partition is refined on \(I_2\): \(\{v_{11},v_{12},v_{13}\}\) forms \(D_3\), while \(\{v_9,v_{10}\}\) remains small in \(H_{\le3}[I_2]\) but is attached to an earlier class through the \(H_3\)-edge \(v_8v_9\), and hence forms \(P_3\). Finally, \(I_3=\{v_{14}\}\). Thus \(S_1=D_1\), \(S_2=D_2\mathbin{\dot\cup}P_2\), \(S_3=D_3\mathbin{\dot\cup}P_3\), and \(S_4=I_3\).}
\label{fig:thiery-exchange-layers}
\end{figure}

% ===============================================
% ===============================================
% \bibliographystyle{plain}
\bibliographystyle{alpha}
\bibliography{manuscript-20261007}
% ===============================================
% ===============================================

\end{document}